\documentclass[letterpaper, 10 pt, final]{IEEEtran}  

\IEEEoverridecommandlockouts                              

\usepackage{amssymb,amsfonts,amsthm}
\usepackage{eucal,bm}
\usepackage{gensymb} 
\usepackage[font=scriptsize]{subcaption}
\usepackage{booktabs}
\usepackage{times,url}
\usepackage{comment}
\usepackage[pdftex,dvipsnames]{xcolor}

\usepackage{cite}
\usepackage{graphicx}
\usepackage[cmex10]{amsmath}
\usepackage{epstopdf}
\usepackage{breqn}

\newcommand{\f}[2]{\frac{#1}{#2}}
\newcommand{\pd}[2]{\f{\partial #1}{\partial #2}}
\renewcommand{\d}[2]{\frac{d #1}{d #2}} 
\DeclareMathOperator*{\argmax}{arg\,max}
\DeclareMathOperator*{\argmin}{arg\,min}

\newtheorem{theorem}{Theorem}
\newtheorem{lemma}{Lemma}

\newtheorem{assumption}{Assumption}
\newtheorem{remark}{Remark}
\newtheorem{proposition}{Proposition}

\allowdisplaybreaks

\title{\LARGE \bf
  Computing Safety Margins of Parameterized Nonlinear Systems for
  Vulnerability Assessment via Trajectory Sensitivities}

\author{Michael W. Fisher\thanks{M. W. Fisher is with the 
ECE Department, University of Waterloo, Waterloo, ON Canada,
michael.fisher@uwaterloo.edu.}, {\it Member, IEEE}
}

\begin{document}


\maketitle
\thispagestyle{empty}
\pagestyle{empty}

\begin{abstract}
  Physical systems experience nonlinear disturbances which have the potential
  to disrupt desired behavior.
  For a particular disturbance, whether or not the system recovers from the
  disturbance to a desired stable equilibrium point depends on system parameter
  values, which are typically uncertain and time-varying.
  Therefore, to quantify proximity to vulnerability
  we define the safety margin to be the smallest change in
  parameter values from a nominal value such that the system will no longer
  be able to recover from the disturbance.
  Safety margins are valuable
  but challenging to compute as related methods, such as those for robust
  region of attraction estimation,
  are often either overly conservative or computationally intractable for
  high dimensional systems.
  Recently, we developed algorithms to compute safety margins efficiently
  and non-conservatively by exploiting the large sensitivity of the system
  trajectory near the region of attraction boundary to small perturbations.
  Although these algorithms have enjoyed empirical success, they lack
  theoretical guarantees that would ensure their generalizability.
  This work develops a novel characterization of safety margins in terms of
  trajectory sensitivities, and uses this to
  derive well-posedness and convergence guarantees for these
  algorithms, enabling their generalizability and successful application
  to a large class of nonlinear systems.
\end{abstract}




\section{Introduction}

Physical systems experience nonlinear
disturbances which have the
potential to disrupt desired operation.
For example, a short circuit in a power system can lead to blackout conditions,
and a push on the torso of a humanoid robot can cause it to fall over.
Consider a particular finite-time disturbance, such as a particular short
circuit in a power system or push on a robot.
If the system is able to recover after the disturbance to a desired
stable equilibrium point (SEP), then it is not vulnerable to that disturbance.
Knowledge of the proximity to vulnerability is important because it provides
a quantitative measure of the margins for safe operation.
As disturbance recovery depends on system parameter values, which are typically
uncertain and time-varying, it is valuable to measure proximity to vulnerability
as a function of system parameters. 


Towards that end,
we define the {\it recovery region} to be the set of parameter values for
which the system recovers to the desired SEP, and
the {\it recovery boundary} to be the boundary in
parameter space of the recovery region.
Given an initial nominal parameter value, we call the distance from this
nominal value to the recovery boundary the {\it safety margin}.
The safety margin quantifies proximity to vulnerability by providing a
measure of how much conditions would have to change, or how much
uncertainty can be tolerated, before
the system would become vulnerable to the potential disturbance.
Determining the safety margin is a very challenging problem
because of the intractability of computing the recovery boundary in high
dimensional parameter space.

The recovery boundary depends on the disturbance, which we model
as a map from parameter values to a post-disturbance
initial condition. 
This map can be implicit, such as a known dynamic model for a particular
contingency,
or explicit, such as a static algebraic map, and is therefore a very general
model.
As a special case, consider a constant additive disturbance to the state
of the form $x_0 + p$, where $x_0$ is some nominal state and $p$ are parameters.
In this case, the recovery boundary in parameter space would be equal to
the region of attraction (RoA) boundary of the SEP in state space, and the
safety margin would be the distance from the nominal state to
the RoA boundary.
So, the recovery boundary framework can be thought of as a generalization of
the RoA boundary to incorporate the impact of parameter variation on safety.

Although the notions of recovery region and recovery boundary are novel, there
is a long history with an extensive literature on RoA estimation in state space.
Some of the most common approaches to RoA estimation include methods based on
sublevel sets of Lyapunov functions \cite{Ch88,Ta08,Kh02} and
the solutions of related partial differential equations \cite{El17,Zu64,Va85}.
However, often these methods are either overly conservative in practice or
computationally intractable for high dimensional systems.
There exist many approaches for robust RoA estimation in the
presence of uncertainty \cite{Ca01,Pa98}, including for
parametric uncertainty \cite{To10,Ch13}.
However, these methods are usually derived by extending approaches for RoA
estimation in the absence of uncertainty to this more general setting and,
thus, are prone to inheriting similar properties of conservatism and/or
computational complexity.

To address these limitations, we have developed novel algorithms for
numerically computing safety margins which are both non-conservative and
computationally efficient for high dimensional nonlinear systems \cite{Fi24}.
In particular, in \cite{Fi24} we successfully apply these algorithms to
a realistic, nonlinear power system model with state space dimension 312
and parameter space dimension 86 to compute the safety margins in
response to a short circuit to within a chosen tolerance of $10^{-3}$.
This requires only 39 iterations, each of which involves a time domain
simulation and the solution of a quadratic program.
However, despite the empirical success of these algorithms, they lack
theoretical guarantees such as well-posedness and convergence certificates
that would ensure their generalizability to a wide variety of nonlinear systems
and operating conditions.

The purpose of the present paper is to address this gap by providing
well-posedness and convergence guarantees for these
algorithms for a large class of parameterized nonlinear systems.
To do so, we begin by noting that the derivative of the system trajectory at
each time with respect to parameters, known as the trajectory sensitivity,
can be efficiently computed numerically from a single time domain
simulation \cite{Hi00}.
It seems intuitive that the trajectory sensitivities will be very small for
post-disturbance initial conditions near the SEP, and will diverge towards
infinity as the RoA boundary is approached.
Motivated by this intuition, we define a function $G$ to be the reciprocal of
the supremum over time of the norm of the trajectory sensitivity at each time.
Under mild assumptions about the structure of the RoA boundary - similar to
those in \cite{Ch88} - we prove that for a parameter value $p$ in the recovery
boundary almost everywhere, $G(p) = 0$ and $G$ extends to a $C^2$ function on
a neighborhood of $p$.
Furthermore, we show that $G$ is
strictly positive and continuous over the recovery region.
Thus, we show that $G$ provides an exact characterization of the recovery
boundary in terms of trajectory sensitivities.

This characterization allows us to transform the challenging and abstract
problem of computing safety margins into a concrete numerical optimization
problem where we seek the closest parameter value $p^*$ to a nominal value
subject to the constraint that $G(p^*) = 0$.
More generally, in \cite{Fi24} algorithms are presented which use $G$ to
(i) find a point
on the recovery boundary in one dimensional parameter space, (ii) numerically
trace the recovery boundary in two dimensional parameter space, and (iii) find
the closest point on the recovery boundary to a nominal value in arbitrary
dimensional parameter space to compute the safety margin.
The optimization problem of (iii) is nonconvex, can possess multiple solutions,
and has a solution that can vary nonsmoothly with respect to the nominal
parameter value, so it may not be well-posed and is particularly challenging
to solve.

In the present paper, for an initial parameter value sufficiently close to the
recovery boundary, we provide well-posedness and convergence guarantees for
algorithms (i) and (ii), show that there exists a unique solution to the
nonconvex optimization problem of (iii) that depends smoothly on the nominal
value,
and provide convergence guarantees for algorithm (iii).
It is important to note that these convergence guarantees are for convergence
to the true recovery boundary, not to an estimate or approximation of it, and
therefore represent non-conservative guarantees for these methods of
finding the recovery boundary or the closest point on it.
These theoretical guarantees facilitate successful application
of these algorithms for computing the recovery boundary or the safety margin to
a large class of nonlinear systems.



The remainder of the paper is organized as follows.
Section~\ref{sec:back} provides background information,
Section~\ref{sec:ex} shows a motivating example,
Section~\ref{sec:res} provides the main results,
Sections~\ref{sec:proofs1} and \ref{sec:proofs2} show the proofs,
and Section~\ref{sec:conc} offers concluding remarks.

\section{Background}\label{sec:back}

Let $J$ be a connected smooth manifold that is a subset of some Euclidean space
and represents parameter space.
Let $P$ be a symmetric positive definite matrix, and define the inner
product $\langle p,q \rangle_P = p^\intercal Pq$,
metric $d_p(p,q)^2 = \langle p-q,p-q \rangle_P$,
and norm $||p||_P^2 = \langle p,p \rangle_P$.
For any $p \in J$ and $Q \subset J$, define the set distance
$d_P(p,Q) = \inf_{q \in Q} d_P(p,q)$.
Let $M = \mathbb{R}^n$ be Euclidean space and represent state space.
Let $\{V_p\}_{p \in J}$ be a family of $C^r$ vector fields over $M$ for $r \geq 1$
which vary $C^r$ continuously with parameter.\footnote{In particular, they vary
continuously with respect to the strong $C^r$ topology (see \cite{Hi76} for a
rigorous definition and further background).}
We can define a composite vector field $V$ on $M \times J$ by
$V(x,p) = (V_p(x),0)$ and write $\dot{x} = V(x,p)$.
Then $V$ is $C^r$ so it possesses a $C^r$ flow $\phi$, where
$\phi_{(t,p)}(x)$ denotes the flow from the initial condition $x$ at time $t$
for parameter value $p$.

Consider a fixed $\hat{p} \in J$.
For any $x \in M$ we say that $x$ is nonwandering under $V_{\hat{p}}$ if for any
open neighborhood $U$ of $x$ and any $T \in \mathbb{R}$,
there exists $t \in \mathbb{R}$ of the same sign as $T$ with $|t| > T$ such
that $\phi_{(t,\hat{p})}(U) \cap U \neq \emptyset$.
Let $\Omega(V_{\hat{p}})$ denote the set of all nonwandering points of
$V_{\hat{p}}$, and note that this includes all equilibria and periodic orbits of
$V_{\hat{p}}$, as well as chaotic dynamics and other forms of recurrent behavior.
Let a {\it critical element} $X_{\hat{p}}$ be either an equilibrium point or
a periodic orbit of $V_{\hat{p}}$.
An equilibrium point $X_{\hat{p}}$ of $V_{\hat{p}}$ is hyperbolic if its
linearization $\pd{V_{\hat{p}}(X_{\hat{p}},\hat{p})}{x}$ has no purely imaginary eigenvalues.
A periodic orbit $X_{\hat{p}}$ of $V_{\hat{p}}$ is hyperbolic if there
exists $x \in X_{\hat{p}}$ and
a cross section $S$ containing $x$ such that the Poincare first return
map $\tau:S \to S$ is well-defined and its linearization $\d{\tau(x)}{x}$ has
no eigenvalues of norm one.
Every hyperbolic critical element $X_{\hat{p}}$ of $V_{\hat{p}}$ possesses a
stable manifold
$W^s(X_{\hat{p}})$ and an unstable manifold $W^u(X_{\hat{p}})$ where
$W^s(X_{\hat{p}})$ $\left(W^u(X_{\hat{p}})\right)$ consists of all initial
conditions that converge to
$X_{\hat{p}}$ in forwards (backwards) time.
Furthermore, there exist local stable and unstable manifolds,
denoted $W^s_{\text{loc}}(X_{\hat{p}})$ and $W^u_{\text{loc}}(X_{\hat{p}})$ and which
are invariant
in forwards and backwards time, respectively.
Furthermore, for any hyperbolic critical element $X_{\hat{p}}$ of $V_{\hat{p}}$,
for $J$ sufficiently small and any $p \in J$ there exists a unique hyperbolic
critical element $X_p$ that is equal to $X_{\hat{p}}$ at $p = \hat{p}$, and
such that $X_p$, $W^s_{\text{loc}}(X_p)$, and
$W^u_{\text{loc}}(X_p)$ all vary $C^r$ continuously with $p$.
For any $J' \subset J$, define
$X_{J'} = \cup_{p \in J'} X_p \times \{p\} \subset M \times J$.
Define $W^s(X_{J'})$ and $W^u(X_{J'})$ analogously.


The notion of a generic property is meant to capture typical or expected
behavior, such as if a vector field or parameter were chosen at random, and
can be used to rule out undesirable pathological behavior.
In a measure space, we say a property is generic if it holds for a set of
full measure (i.e., its complement has measure zero).
In a topological space, we say a property is generic if it holds for a
countable intersection of open dense sets.
In this paper, the topological notion of generic is used when referring to
properties of vector fields (for which there is a natural topology, but no
natural measure), and the measure space notion is used when referring to
parameters (which lie in the smooth manifold $J$, where the notion of measure
zero is well-defined).

For any set $A$ in a topological space, let $\partial A$ denote its topological
boundary, and let $\overline{A}$ denote its topological closure.
For any sets $A, B$ in a metric space with metric $d$, define the set distance
$d(A,B)$ to be the infimum over $x \in A$ and $y \in B$ of $d(x,y)$.
For any point $x$ in a metric space and any $r > 0$, define $B_r(x)$ to
be the open ball of radius $r$ centered at $x$.
For any manifold $A$ and $x \in A$, let $T_xA$ denote the tangent space to
$A$ at $x$.
Two manifolds $A, B \subset M$ are transverse if for every $x \in A \cap B$,
$T_xA + T_xB = T_xM$.
For additional background in differential topology, including regular values,
embedded submanifolds, tubular neighborhoods, and retractions, we refer the
reader to \cite[Chapters 4-6]{Lee13}.

Suppose $p_0 \in J$ such that $V_{p_0}$ possesses a hyperbolic
stable equilibrium point (SEP) $X^s_{p_0}$.
Then for $J$ sufficiently small and any $p \in J$, there
exists a unique hyperbolic SEP $X^s_p$ near $X^s_{p_0}$.
Note that for $p \in J$,
$W^s(X^s_p)$ is the RoA of $X^s_p$, and $\partial W^s(X^s_p)$
is the RoA boundary.

We model a finite-time nonlinear parameter-dependent disturbance
as a parameter-dependent post-disturbance initial condition (which we simply
refer to as the initial condition) $y:J \to M$ given by $y_p := y(p)$.
Then, the subsequent dynamics are given by the vector field $V_p$ for
each $p \in J$.
For a parameter value $p \in J$, the system recovers from the
disturbance if and only if $y_p \in W^s(X^s(p))$.
Let the {\it recovery region} $R$ be the set of parameter values in $J$ for
which the system recovers, and define the {\it recovery boundary} to be its
topological boundary $\partial R$ in $J$.
For a fixed $p_0 \in R$, the distance $d(p_0,\partial R)$ is the recovery
margin, and provides a quantitative measure of the margins for safe operation.
Therefore, it is valuable to develop algorithms for finding points on the
recovery boundary - often, the closest such point - in order to determine
the safety margin.

Proposition~\ref{thm:multi} was proved in \cite{Fi23,Fi22}, and states that
under general assumptions satisfied
by a large class of dynamical systems, and for
sufficiently small $J$, the following holds.
The RoA boundary in state space varies continuously with respect to parameter,
and is equal to the union of the stable manifolds of the critical elements
it contains.
Every parameter value in the recovery boundary has corresponding initial
condition which lies on the RoA boundary in state space, and therefore in the
stable manifold of some critical element, which we call the controlling
critical element.
The unstable manifold of the controlling critical element intersects the RoA.
This proposition provides an explicit link between the recovery boundary in
parameter space and the RoA boundary in state space,
which will be exploited for the design and theoretical guarantees of
algorithms which compute points on the recovery boundary.

\begin{proposition}
  \label{thm:multi}
  \cite[Theorem 4.17, Corollary 4.18, Theorem 4.21, Corollary 4.23]{Fi23}
Assume there exists $p_0 \in J$ such that:
\begin{itemize}
\item[(i)] Every critical element in $\partial W^s(X^s_{p_0})$ is hyperbolic.
\item[(ii)] The intersections of the stable and unstable manifolds of the
  critical elements in $\partial W^s(X^s_{p_0})$ are transverse.
\item[(iii)] The intersection of $\partial W^s(X^s_{p_0})$ with $\Omega(V_{p_0})$
  consists of a finite union of critical elements $\{X^i_{p_0}\}_{i \in I}$.
\item[(iv)] There exists a neighborhood of infinity which contains no
  nonwandering points of $V_{p_0}$ and no orbits which diverge to infinity in
  both forwards and backwards time.
\item[(v)] Additional generic assumptions.\footnote{The additional generic
  assumptions are that certain
  lower semicontinuous functions over $J$ are continuous at $p_0$
  (see \cite[Theorem~4.17]{Fi23} for more details).}
\end{itemize}
Then for $J$ sufficiently small, $\partial W^s(X^s_p)$ varies continuously with
$p$,\footnote{The continuity is with respect to the Chabauty topology, which
  is an extension of the Hausdorff topology to (noncompact) Euclidean space.}
$\partial W^s(X^s_J) = \bigcup_{i \in I} W^s(X^i_J)$,
and for any $p^* \in \partial R$, $y_{p^*} \in \partial W^s(X^s_{p^*})$.
Then $y_{p^*} \in W^s(X^j_{p^*})$ for a unique $j \in I$, we call
$X^j_{p^*}$ the controlling critical element for $p^*$, and
$W^u(X^j_{p^*}) \cap W^s(X^s_{p^*}) \neq \emptyset$.
\end{proposition}

Note that Assumptions (i), (ii), and (v) are generic, and it is generically
true that $\Omega(V_{p_0})$ is equal to the closure of the union of the
critical elements of $V_{p_0}$
(compare to Assumption (iii)), so the conditions of
Proposition~\ref{thm:multi} are very
general and hold for a large class of realistic engineering system models.

As $\phi_{(t,p)}(y(p))$ gives the flow of the vector field $V_p$
at time $t$ starting from initial condition $y(p)$,
as $t$ varies from zero to infinity this traces out the system trajectory
after the disturbance.
Then $\d{\phi_{(t,p)}(y(p))}{p}$ is the derivative of this flow with respect to
parameter, and represents the sensitivity of the post-disturbance system
trajectory with respect to parameter: the trajectory sensitivity.
The trajectory sensitivities are a function of parameter $p$ and time $t$.
We note that trajectory sensitivities can be efficiently computed numerically
as a byproduct of the underlying numerical integration of the system dynamics
\cite{Hi00}.

We will use local linearizations near an equilibrium point to
analyze the behavior of the trajectory sensitivities.
However, since trajectory sensitivities require a derivative of the flow,
standard $C^0$ linearizations (such as that provided by the Hartman-Grobman
Theorem) will not preserve the trajectory sensitivities in the local
coordinates.
Therefore, we will need to resort to smooth linearizations \cite{Se85},
which require an additional assumption: the strong Sternberg condition.
We define an integer combination of order $Q$ to be a linear combination
with nonnegative integer coefficients where the sum of the coefficients is
equal to $Q$.
We say that a matrix $N$ 
satisfies the strong Sternberg condition
of order $Q$ if (i) every integer combination of its eigenvalues of order
between $2$ and $Q$ is not equal to any single eigenvalue, and
(ii) the real part of every integer combination of its eigenvalues of order $Q$
is not equal to the real part of any single eigenvalue.
This is a mild non-resonance assumption.
Define $\rho^-$ to the ratio of the maximum over the minimum of the absolute
values of the real parts of the stable eigenvalues of $N$, and define
$\rho^+$ analogously for the unstable eigenvalues.
Then for any positive integer $Q$, the $Q$-smoothness of $N$ is the largest
integer $K \geq 0$ such that $Q = Q^+ + Q^-$ for $Q^+, Q^- \geq 0$ integers,
$Q^+ \geq K\rho^+$, and $Q^- \geq K\rho^-$.  Note that the $Q$-smoothness
approaches infinity as $Q$ does, so for $Q$ sufficiently large it will be
at least three.
For any complex-valued matrix $N$, let $||N||_1 = \sum_{i,j} |N_{i,j}|$.
Let $\text{diag}(A_1,...,A_n)$ refer to the block diagonal matrix with
matrices $A_1$, ..., $A_n$ on the diagonal.

\section{Motivating Example}\label{sec:ex}

To motivate the algorithms of this paper, we provide intuition from
the simple power system model of a single machine infinite bus, which has the
same structure as a nonlinear pendulum as follows
\begin{align*}
  \dot{x}_1 &= x_2 \\
  \dot{x}_2 &= p_1 \sin(x_1) - 0.5 x_2 + p_2
\end{align*}
where $p_1$ and $p_2$ are parameters.
The disturbance we consider is a temporary short circuit which is modeled by
temporarily setting $p_1 = 0$ for a time of $0.8$ seconds and then restoring
$p_1$ to its prior value.
We define the reciprocal trajectory sensitivity function $G$ to be
\begin{align*}
  G(p) = \inf_{t \geq 0} \left|\left|\pd{\phi_{(t,p)}(y(p))}{p}\right|\right|_1^{-1}.
\end{align*}

Figure~\ref{fig:roa}(a) shows the recovery region and recovery
boundary\footnote{In this simple example the recovery boundary is linear, but
this is typically not the case for more complex systems.} in
parameter space for
this disturbance, and the safety margin for a nominal parameter value of
$\begin{bmatrix} p_1 & p_2 \end{bmatrix} =
\begin{bmatrix} 1.9 & 1.5 \end{bmatrix}$.
Figure~\ref{fig:roa}(b) shows the (post-disturbance) initial condition and
the RoA boundary of the SEP in state space for three parameter values shown in
Figure~\ref{fig:roa}(a).
For the parameter values inside the recovery region, on the recovery boundary,
and outside the recovery boundary, their initial conditions lie inside
the RoA, on the RoA boundary, and outside the RoA boundary, respectively.
Figure~\ref{fig:traj}(a) shows a contour plot of 
$G$ in parameter space.  Note that $G$ converges to zero
exactly at the true recovery boundary.\footnote{In this simple example $G$
approaches zero monotonically as the recovery boundary is approached, but this
is not always the case for more complex systems.}
Figure~\ref{fig:traj}(b) plots a representative trajectory sensitivity as
a function of time for three parameter values shown in Figure~\ref{fig:traj}(a).
Note that the trajectory sensitivities grows large as parameter values approach
the recovery boundary, providing intuition for the contour plot of $G$ shown
in Figure~\ref{fig:traj}(a).


\begin{figure}[t!]
  \centering
  \begin{subfigure}[t]{0.49\linewidth}
    \centering
    \includegraphics[width=\linewidth]{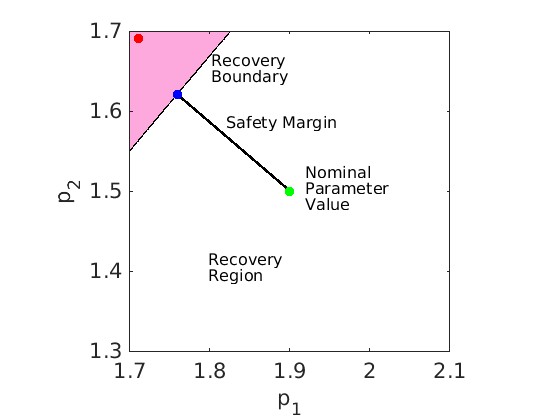}
    \caption{The recovery region, recovery boundary, safety margin,
      and three highlighted parameter values (open circles).}
  \end{subfigure}%
  \hfill
 \begin{subfigure}[t]{0.49\linewidth}
   \centering
   \includegraphics[width=\linewidth]{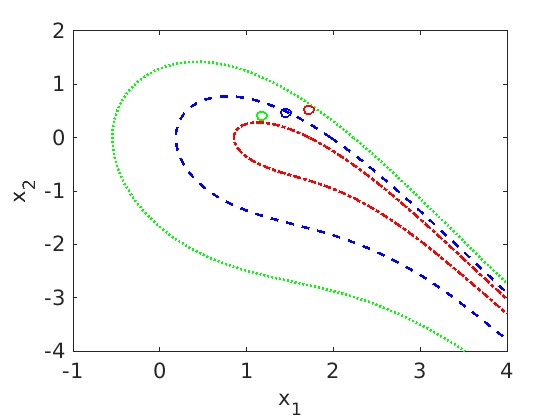}
   \caption{The post-disturbance initial conditions (open circles) and the
   RoA boundaries (dotted and dashed lines) corresponding by color to the
   highlighted parameter values in $(a)$.}
 \end{subfigure}%
 \caption{Parameter space and state space for the example of
   Section~\ref{sec:ex}.}
 \label{fig:roa}
\end{figure}

\begin{figure}[t!]
  \centering
  \begin{subfigure}[t]{0.47\linewidth}
    \centering
    \includegraphics[width=\linewidth]{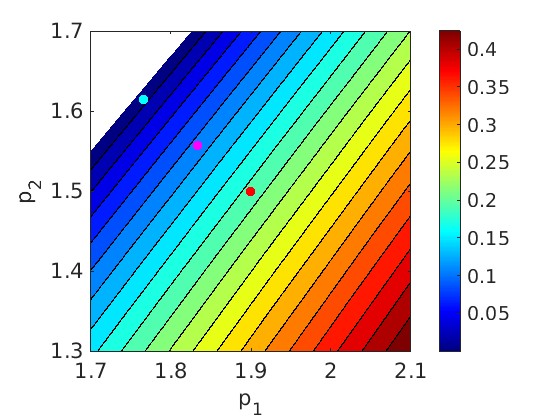}
    \caption{Contour plot for the trajectory sensitivity function $G$
      and three highlighted parameter values.}
  \end{subfigure}%
  \hfill
 \begin{subfigure}[t]{0.51\linewidth}
   \centering
   \includegraphics[width=\linewidth]{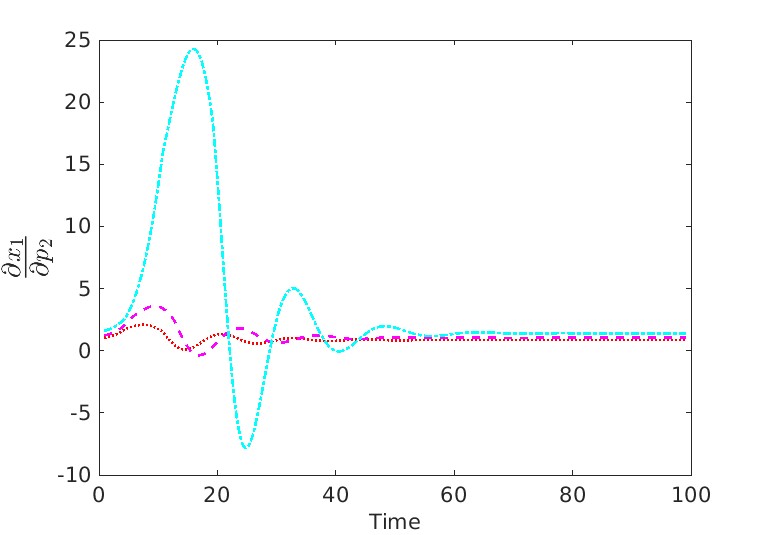}
   \caption{The trajectory sensitivity of state $x_1$ with respect to
     parameter $p_2$ as a function of time for the highlighted parameter
     values in (a).}
 \end{subfigure}%
 \caption{Trajectory sensitivities for the example of Section~\ref{sec:ex}.}
 \label{fig:traj}
\end{figure}

\section{Main Results}\label{sec:res}

\subsection{Expressing Recovery Boundary via Trajectory Sensitivities}

The recovery boundary $\partial R$ is defined abstractly, is typically
nonsmooth and nonconvex, and is usually intractable to compute or
approximate accurately, especially in high dimensions.
The first main result of this work is to provide an exact characterization of
the recovery boundary in terms of trajectory sensitivities, which measure the
derivative of the flow with respect to parameters and can be
efficiently computed numerically.  This characterization will then be used
to transform the abstract problem of finding points on the recovery boundary
into a concrete formulation which can be solved using tractable and efficient
numerical algorithms \cite{Fi24,Fi18c}, as discussed further in
Section~\ref{sec:algos}.
Before presenting the main result of this section, several assumptions and
definitions are required.

Assume the conditions of Proposition~\ref{thm:multi}.
Then for each $p^* \in \partial R$ there exists
a controlling critical element $X^*(p^*)$ such that $y(p^*) \in W^s(X^*(p^*))$.
All of the assumptions below are made for $J$ sufficiently small; in other
words, they are local assumptions in parameter space.

\begin{assumption}
  \label{as:yj}
By Proposition~\ref{thm:multi},
$\partial W^s(X^s_J) = \bigcup_{i \in I} W^s(X^i_J)$.
Assume that $y_J$ is transverse to $W^s(X^i_J)$ for all $i \in I$.
\end{assumption}

\begin{remark}
Assumption~\ref{as:yj} is generic because $C^1$ submanifolds
are generically transverse \cite[Theorem A.3.20]{Ka99}.
\end{remark}

\begin{assumption}
  \label{as:init}
  Assume that for each $p \in R$, there exists some time $t \geq 0$ such
  that $\d{\phi_{(t,p)}(y(p))}{p} \neq 0$.
\end{assumption}

\begin{remark}
 Assumption~\ref{as:init} implies that the parameter $p$ has a nonzero effect
 on the flow, which is a very mild assumption.
\end{remark}

\begin{assumption}
  \label{as:smooth}
The vector field family $\{V_p\}_{p \in J}$ is strong $C^3$ continuous,
each vector field $V_p$ for $p \in J$ is $C^\infty$ at each equilibrium point
in $\overline{W}^s(X^s(p))$, and $y$ is $C^3$.
\end{assumption}

\begin{remark}
  Assumption~\ref{as:smooth} is required to ensure the existence of a smooth
  local linearization (see Assumption~\ref{as:res}) near equilibrium points.
  It is possible to relax the condition that each vector field be $C^\infty$
  at equilibria to $C^Q$ for a particular finite $Q$, but for simplicity of
  presentation we focus on the $C^\infty$ case.
  It may be possible in future work to relax the assumption of $C^3$ continuity
  to certain classes of piecewise smooth vector fields
  (see, e.g., \cite{Fi19b}).
\end{remark}


\begin{assumption}
  \label{as:res}
  For any $p \in \overline{R}$,  if $p \in R$ let $X = X^s(p)$ and if
  $p \in \partial R$ let $X = X^*(p)$.
 Assume that the matrix $\pd{V_p(X)}{x}$
 satisfies the strong Sternberg condition for $Q$ with $Q$ sufficiently large
 such that the $Q$-smoothness is at least three.
\end{assumption}

\begin{remark}
The strong Sternberg condition for $Q$ of Assumption~\ref{as:res}
is a non-resonance condition that will ensure the existence of a $C^3$
local linearization about the controlling critical element.
Note that this is stronger than the $C^0$ local linearization given by
the Hartman-Grobman Theorem, and is required to ensure that the trajectory
sensitivities in the local linear coordinates
can be related to those of the original nonlinear system.
\end{remark}

\begin{assumption}
  \label{as:oneD}
For every $p^* \in \partial R$, 
if $W^s(X^*(p^*))$ has codimension one then $X^*(p^*)$ is an equilibrium point.
\end{assumption}

\begin{remark}
For many practical nonlinear systems including power systems,
Assumption~\ref{as:oneD} is ubiquitous and
has been observed through numerical experiments on a wide range of
realistic models \cite{Tr96}.
\end{remark}

For any $p^* \in \partial R$ where $W^s(X^*(p^*))$ has codimension one,
since by Assumption~\ref{as:oneD} $X^*(p^*)$ is an equilibrium point, this
implies by hyperbolicity that the unstable manifold $W^u(X^*(p^*))$ is one
dimensional.

\begin{assumption}
  \label{as:unique}
Let $p^* \in \partial R$ such that $W^s(X^*(p^*))$ has codimension one.
We will define a function $f(t,p)$ that is equal to
$g(p)\left|\left|\d{\phi_{(t,p)}(y(p))}{p}\right|\right|_1$, where $g(p)$ is
a particular scalar function, in the limit as $p \to p^*$.
We will show that $f(t,p^*)$ attains a finite maximum over $t \geq 0$.
Assume that this maximum is attained at a unique time $\tilde{t}$.
\end{assumption}

\begin{remark}
  Assumption~\ref{as:unique} is required to ensure the trajectory
  sensitivities vary $C^2$ with parameter near $\partial R$.  It could be
  relaxed to attaining the maximum at a finite number of times.
\end{remark}

Define a function $H:[0,\infty) \times \overline{R} \to [0,\infty]$
by $H(t,p) = \left|\left|\d{\phi_{(t,p)}(y(p))}{p}\right|\right|_1^{-1}$.
Then $H(t,p)$ represents the reciprocal of the norm of the trajectory
sensitivity at time $t$ and for parameter value $p$.
Define the function $G:\overline{R} \to [0,\infty)$ by
$G(p) = \inf_{t \in [0,\infty)} H(t,p)$.
Then $G$ gives the infimum over time of the reciprocal of the norm of the
trajectory sensitivity.

Intuitively, it is natural to expect that the flow would become infinitely
sensitive to perturbations along the RoA boundary, and less sensitive from
within the RoA.
This would imply that the trajectory sensitivities would diverge to infinity
along the RoA boundary and stay finite within the RoA, and that the reciprocal
trajectory sensitivities would approach zero along the RoA boundary and stay
positive within the RoA.
By Proposition~\ref{thm:multi}, $p^* \in \partial R$ implies that the initial
condition $y(p^*)$ lies on the RoA boundary.
Therefore, it seems intuitive that $p^* \in \partial R$ would imply that
$G(p^*) = 0$, and that $p \in R$ would imply that $G(p) > 0$.
Theorem~\ref{thm:Gtot} makes this intuition explicit, thus providing a
characterization of
$\overline{R}$ in terms of the reciprocal trajectory sensitivities $G$.
Furthermore, it shows that $G$ is continuous so that $G(p)$ approaches zero as
$p$ approaches $\partial R$, motivating algorithms which find points on
$\partial R$ by solving for roots of $G$ \cite{Fi24,Fi18c}.
It also shows that $G$ extends to a $C^2$ function near $\partial R$, which
is valuable for providing convergence guarantees for these algorithms, as
presented in Section~\ref{sec:algos}.

\begin{theorem}
  \label{thm:Gtot}
Assume the conditions of Proposition~\ref{thm:multi}
and Assumptions~\ref{as:yj}-\ref{as:unique}.
Then $G$ 
is finite, strictly positive, and 
continuous over $R$.
For generic\footnote{By Lemma~\ref{lem:Rmanifold}, $\partial R$ consists of
  a finite union of $C^3$ submanifolds.
  The generic assumption is that $p^*$ belongs to a submanifold with
  codimension one.} $p^* \in \partial R$, 
$G(p^*) = 0$, $G$ is continuous at $p^*$, and
$G$ extends
to a $C^2$ function over an open
neighborhood of $p^*$ in $J$
which is strictly negative outside of $\overline{R}$.
\end{theorem}

\begin{remark}
  If all critical elements in the RoA boundary are equilibria,
  it can be shown that $G(p^*) = 0$ for all
  $p^* \in \partial R$, so that $\partial R = G^{-1}(0)$.
  We conjecture that the same holds even in the presence of periodic
  orbits in the RoA boundary.
\end{remark}

\subsection{Recovery Boundary Algorithms and Convergence Guarantees}
\label{sec:algos}

Based on the characterization of the recovery boundary in terms of trajectory
sensitivities from Theorem~\ref{thm:Gtot}, we present algorithms that minimize
inverse trajectory sensitivities to find points on the recovery boundary, and
provide convergence guarantees for these algorithms.
We note for use in these algorithms that for any $p \in R$, $G(p)$ and $DG(p)$
can be efficiently computed numerically from a single time domain simulation
\cite{Fi24,Fi18c}.
The algorithms are organized based on the dimension of the parameter space
$J$ in which the points on the recovery boundary are computed, which is
determined by the user based on their selection of the parameters of interest.

\subsubsection{Finding Recovery Boundary in One Dimensional Parameter Space}

When the parameter space $J$ is one dimensional,
a connected component of the recovery boundary
consists of a single point.
Therefore, in this setting our goal is to find a single point on the recovery
boundary, i.e., $p^* \in \partial R$.
To do so, motivated by Theorem~\ref{thm:Gtot} we wish to solve for $p^*$
which satisfies $G(p^*) = 0$.
To solve this nonlinear equation, we would like to apply Newton-Raphson,
which results in the following update at each iteration $s$:
\begin{align}
  \tilde{F}(p^s) := p^s - DG(p^s)^{-1}G(p^s).
  \label{eq:1dim}
\end{align}
However, it is possible that for $p^s \in R$, $\tilde{F}(p^s) \not\in R$.
As the asymptotic behavior of the system is not known outside of $R$, where
the system
will not recover from the disturbance, and since $G$ is not defined outside
of $\overline{R}$, it is desirable to enforce that each iteration of the
algorithm lies inside $R$.
This is achieved with a backtracking line search performed by bisection
which is defined formally as follows. Let $m \in \{0, 1, ...\}$,
$\tilde{F}$ be as in \eqref{eq:1dim}, and define
\begin{align}
  \begin{split}
  F_m(p^s) &= p^s + \f{1}{2^m}\left(\tilde{F}(p^s)-p^s\right) \\
  m(p^s) &= \min\{m \in \{0, 1, ...\}:F_m(p^s) \in R\} \\
  p^{s+1} &= F(p^s) := F_{m(p^s)}(p^s).
  \label{eq:backtrack}
  \end{split}
\end{align}
As $R = \{p:G(p) > 0\}$ is open by continuity of $G$, for any $p^s \in R$,
$m(p^s)$ is finite.
Thus, $p^{s+1} \in R$ so $R$ is forward invariant under this algorithm.

Theorem~\ref{thm:1dim} provides convergence guarantees for finding a point
on the recovery boundary in one dimensional parameter space using the algorithm
defined by \eqref{eq:backtrack} with $\tilde{F}$ as in \eqref{eq:1dim}.

\begin{theorem}
  \label{thm:1dim}
Assume the conditions of Theorem~\ref{thm:Gtot}
and an additional generic assumption.\footnote{The generic assumption is that
zero is a regular value of the $C^2$ function $\tilde{G}$.}
Then there exists an open neighborhood $N$ of $\partial R$ such that for
$p_0 \in N$, the sequence $\{p^s\}_{s=1}^\infty$ starting
from $p^1 = p_0$ and defined by the algorithm \eqref{eq:backtrack} with
$\tilde{F}$ as given by \eqref{eq:1dim} is well-defined and converges to a
unique $p^* \in \partial R$.
\end{theorem}



\subsubsection{Tracing Recovery Boundary in Two Dimensional Parameter Space}

When the parameter space $J$ is two dimensional, by the proof of
Theorem~\ref{thm:Gtot} (in particular, Lemma~\ref{lem:Rmanifold}) the recovery
boundary consists of a finite union of
one dimensional manifolds (i.e., smooth curves).
Therefore, in this setting our goal is to
numerically trace the recovery boundary by generating a sequence of points
along it.  To do so, we use a predictor-corrector algorithm
\cite{Fi24} which alternates between predictor steps, which start
from a point on $\partial R$ and move along the tangent to
$\partial R$ to find a predicted point, and corrector steps, which
project onto a point in the intersection of $\partial R$ with the hyperplane
containing the predicted point and orthogonal to that tangent.
More formally, for $\hat{p} \in \partial R$ the predicted point is given by
\begin{align}
  p^{pred}(p) = \hat{p} + \kappa \eta(p) \label{eq:pred}
\end{align}
where $\kappa > 0$ and $\eta(p)$ is the tangent to $\partial R$ at $p$.
Next, the hyperplane for the correction step is equal to the set of $p \in J$
that satisfy
\begin{align*}
  (p-p^{pred}(p))^\intercal(p^{pred}(p)-\hat{p}) = 0
\end{align*}
which is equivalent by \eqref{eq:pred} to
\begin{align*}
  (p-\hat{p})^\intercal \eta(p) - \kappa = 0,
\end{align*}
so let $H_\kappa$ denote this hyperplane.
Thus, to find a point in the intersection of $H_\kappa$ and $G^{-1}(0)$,
it suffices to solve
\begin{align*}
  f(p) := \begin{bmatrix}
    G(p) \\
    (p-\hat{p})^\intercal \eta(p) - \kappa
  \end{bmatrix} = 0.
\end{align*}
To solve these nonlinear equations, we use the following Newton-Raphson update
at each iteration $s$:
\begin{align}
  \tilde{F}(p^s) := p^s - Df(p^s)^{-1}f(p^s) \label{eq:2dim}
\end{align}
together with the backtracking line search of \eqref{eq:backtrack} so that
$R$ is forward invariant under this algorithm.
If $p^{pred}(\hat{p}) \in R$ then we set $p^1 = p^{pred}(\hat{p})$.
Otherwise, observing that $H_\kappa$ is a line in $J$, we perform a line
search using, for example, the bisection or golden section search
  methods, to find $p^1 \in H_\kappa \cap R$.

Theorem~\ref{thm:cont} provides convergence guarantees for finding the next
point on the recovery boundary from the current point using the algorithm
defined by \eqref{eq:backtrack} with $\tilde{F}$ as in \eqref{eq:2dim}.
Repeated application of this theorem therefore guarantees that this algorithm
will numerically trace a
sequence of points along the recovery boundary, and can approximate the
recovery boundary to arbitrary accuracy as the step size $\kappa \to 0$.

\begin{theorem}
  \label{thm:cont}
  Assume the conditions of Theorem~\ref{thm:Gtot} and an additional generic
assumption.\footnote{The generic assumption is that zero is a regular value of
  the $C^2$ function $\tilde{G}$.}
Then for generic $\hat{p} \in \partial R$ there exists $r > 0$ such that for
generic $\kappa \in (0,r)$,  the
sequence $\{p^s\}_{s=1}^\infty$ starting from $p^1 \in H_\kappa \cap R$ and
defined by the algorithm \eqref{eq:backtrack} with
$\tilde{F}$ as given by \eqref{eq:2dim} is well-defined and
converges to a unique $p^* \in \partial R \cap H_\kappa$.
\end{theorem}



\subsubsection{Closest Point on Recovery Boundary in Arbitrary Dimensional
  Parameter Space}

For parameter space $J$ of high dimension, it is no longer feasible to
numerically trace the recovery boundary.
Instead, our goal in this setting is to find the closest point on the recovery
boundary to some initial parameter value $p_0 \in J$, which provides a
quantitative measure of the margins for safe operation.
Therefore, we wish to solve the following
(abstract) optimization problem for some $p_0 \in J$ and $P$ symmetric
positive definite:
\begin{align}
  \begin{split}
  \min_p & \quad \f{1}{2}(p-p_0)^\intercal P(p-p_0) \\
  \text{s.t. } &\quad p \in \partial R. \label{eq:opt1}
  \end{split}
\end{align}
Note that the choice of $P$ determines the metric for measuring distance
in parameter space (e.g., $P = I$ results in the Euclidean metric).
This is a very challenging problem to solve, because efficient computational
methods for identifying the constraint set $\partial R$ are not known,
especially in higher dimensional parameter spaces, and $\partial R$ is in
general nonconvex and nonsmooth.
However, Theorem~\ref{thm:Gtot} allows us to transform the abstract
problem of \eqref{eq:opt1} into the following concrete
optimization problem:
\begin{align}
  \begin{split}
  \min_p & \quad \f{1}{2}(p-p_0)^\intercal P(p-p_0) \\
  \text{s.t. }& \quad G(p) = 0. \label{eq:opt2}
  \end{split}
\end{align}
Unfortunately, in general \eqref{eq:opt2}
may not be feasible, there may be many - possibly infinite - solutions,
and the solution(s) may vary discontinuously with the initial parameter value
$p_0$.  These imply that there may not exist a unique closest point on the
recovery boundary for every $p_0$ and, even if there does, that it may
experience discrete jumps as $p_0$ varies.
These present serious potential challenges for numerical algorithms which aim
to solve \eqref{eq:opt2}.

Fortunately, Theorem~\ref{thm:opt} shows that
for generic $p_0$ sufficiently close to the recovery boundary,
$p_0$ has a unique
closest point on the recovery boundary that
varies smoothly with respect to $p_0$.
This ensures that \eqref{eq:opt2}
is a well-posed and well-behaved problem, as it has a unique solution which
varies smoothly with $p_0$,
which will be crucial for proving convergence of the main algorithm of this
section.

\begin{theorem}
  \label{thm:opt}
Assume the conditions of Theorem~\ref{thm:Gtot}.
Then there exists an open neighborhood $N$ of $\partial R$ such that
for generic $p_0 \in N$
there exists a unique solution to
\eqref{eq:opt2}, and that this solution varies $C^3$
with initial parameter value $p_0$.
\end{theorem}


As our goal is to solve \eqref{eq:opt2} numerically, as with any numerical
optimization problem we can only require that the equality constraints
(in this case, $G(p) = 0$) be satisfied up to some finite tolerance
$\epsilon > 0$.
In the following we explicitly introduce this tolerance $\epsilon$ for use
in the convergence analysis.
This leads to the following numerical optimization variation of
\eqref{eq:opt2}:
\begin{align}
  \begin{split}
  \min_{p,x} & \quad \f{1}{2} (p-p_0)^\intercal P(p-p_0) \\
  \text{s.t. }& \quad G(p) = x \\
  & \quad |x| \leq \epsilon.
  \label{eq:num}
  \end{split}
\end{align}
We will show (in Lemma~\ref{lem:num} below) that the solution to this
optimization
problem is given by the following optimization:
\footnote{The only exception is the trivial special case where
  $|G(p_0)| \leq \epsilon$, in which case the solution is $p^* = p_0$.
  This case is not of practical concern.}
\begin{align}
  \begin{split}
  \min_{p,s} & \quad \f{1}{2} (p-p_0)^\intercal P(p-p_0) \\
  \text{s.t. }& \quad G(p) = \epsilon \label{eq:opt3}
  \end{split}
\end{align}

To solve \eqref{eq:opt3}
we consider the sequential quadratic programming algorithm employed in
\cite{Fi24}
which solves the following quadratic program at each iteration $s$
obtained by linearizing the nonconvex
constraint $G(p) = \epsilon$ from \eqref{eq:opt3}: 
\begin{align}
  \begin{split}
  \tilde{F}(p^s) := \argmin_p& \quad \f{1}{2}(p-p_0)^\intercal P(p-p_0) \\
  \text{s.t. }& \quad G(p^s) + DG(p^s)^\intercal(p-p^s) = \epsilon.
  \label{eq:opt4}
  \end{split}
\end{align}
We combine \eqref{eq:opt4} with the backtracking line search of
\eqref{eq:backtrack} so that $R$ is forward invariant
under this algorithm.

Theorem~\ref{thm:conv} provides convergence guarantees for finding the closest
point on the recovery boundary from an initial parameter value $p_0$ using
the algorithm defined by \eqref{eq:backtrack} with $\tilde{F}$ given by
\eqref{eq:opt4}.


\begin{theorem}
  \label{thm:conv}
Assume the conditions of Theorem~\ref{thm:Gtot}.
Then there exists an open neighborhood $N$ of $\partial R$ such that
for generic $p_0 \in N$ there exists a unique
solution $\hat{p}$ to \eqref{eq:opt2} and the following holds.
Under an additional generic assumption,\footnote{The additional generic
  assumption is that $\hat{p}$ lies in a set of full measure in $\partial R$
and zero is a regular value of the $C^2$ function $\tilde{G}$.}
  for any $\epsilon \in (0,G(p_0))$
there exists a unique solution $p^*$ to \eqref{eq:opt3} and the sequence
$\{p^s\}_{s=1}^\infty$ starting from $p^1 = p_0$ and defined by the algorithm
\eqref{eq:backtrack} with $\tilde{F}$ as given by \eqref{eq:opt4} is
well-defined and converges to $p^*$.
\end{theorem}



\section{Proof of Theorem~\ref{thm:Gtot}}
\label{sec:proofs1}

This section contains the proof of Theorem~\ref{thm:Gtot}.
Before providing it, however, a number of technical lemmas are required.
Lemma~\ref{lem:R} will establish the
first claim of Theorem~\ref{thm:Gtot}.

\begin{lemma}
  \label{lem:R}
  $G$ is well-defined, 
  positive, and continuous on $R$.
\end{lemma}

\begin{proof}[Proof of Lemma~\ref{lem:R}]



Define the functions $\mathcal{F}:[0,\infty) \times \overline{R} \to [0,\infty)$
by $\mathcal{F}(t,p) = \left|\left|\d{\phi_{(t,p)}(y(p))}{p}\right|\right|_1$
and $F:\overline{R} \to [0,\infty]$
by $F(p) = \sup_{t \in [0,\infty)} \mathcal{F}(t,p)$.
Then $G = \f{1}{F}$, so to prove the claim it suffices to show that $F$ is
positive, finite, and continuous over $R$.
Fix $\hat{p} \in R$.
First we show that $F$ is positive over $R$.
As $\hat{p}$ is arbitary, it suffices to show that $F(\hat{p}) > 0$.
By Assumption~\ref{as:init}, there exists some time $t \geq 0$ such
that $\d{\phi_{(t,\hat{p})}(y(\hat{p}))}{p} \neq 0$.
This implies that $\mathcal{F}(t,\hat{p}) > 0$, so
$F(\hat{p}) \geq \mathcal{F}(t,\hat{p}) > 0$.


By Assumption~\ref{as:res} and \cite[Theorem 7]{Se85}, there exists a
neighborhood $N$ of $X^s(\hat{p})$, a
neighborhood $J'$ of $\hat{p}$ in $J$, and a $C^3$ conjugacy
$h:N \times J' \to \mathbb{R}^n$ known as a smooth linearization
such that for $p \in J'$, $X^s(p) \in N$ and the vector field $V_p$ is
conjugate by $h$ to a linear vector field in $\mathbb{R}^n$.
In particular, in these coordinates the vector field $V_p$ has the form
$V_p(x) = A_s(p)x$ where $A_s(p)$ is a $C^3$ matrix such that
$||e^{A_s(p)}|| < 1$ for all $p \in J'$.
As $y(\hat{p}) \in W^s(X^s(\hat{p}))$, there exists $t' > 0$ such that
$\phi_{(t',\hat{p})}(y(\hat{p})) \in N$.
Shrink $J'$ if necessary so that $\phi_{(t',p)}(y(p)) \in N$ for all $p \in J'$,
and let $\hat{y}(p) = \phi_{(t',p)}(y(p))$.
Then the dynamics in these coordinates are given by
the flow $\psi_{(t,p)}(\hat{y}(p)) = e^{A_s(p)t}\hat{y}(p)$.
There exist invertible matrices $W(p)$, and block diagonal matrices
$J(p) = \text{diag}(J_1(p), ..., J_m(p))$
where each $J_i(p)$ is an elementary Jordan block of size $n_i$ with eigenvalue
$\lambda_i(p)$, such that $A_s(p) = W(p)J(p)W(p)^{-1}$.
Define
$N_n = {\scriptsize \begin{bmatrix}
  1 & t &  \hdots  & \f{t^{n-1}}{(n-1)!} \\
  &  \ddots & \ddots & \vdots \\
  & & \ddots &   t \\
  & & &   1
  \end{bmatrix}}$, \\
$E(p) = \text{diag}(e^{\lambda_1 t} N_{n_1}, ..., e^{\lambda_n t} N_{n_m})$,
$\Lambda(p) = \text{diag}(\lambda_1 I_{n_1}, ..., \lambda_m I_{n_m})$,
$\d{\Lambda(p)}{p} = \begin{bmatrix}
  \pd{\Lambda(p)}{p_1} & \hdots & \pd{\lambda(p)}{p_{\text{dim }J}} \end{bmatrix}$, and
$\d{W(p)}{p} = \begin{bmatrix} \pd{W(p)}{p_1} & \hdots &
  \pd{W(p)}{p_{\text{dim } J}} \end{bmatrix}$.
Although in general $W(p)$ and $\Lambda(p)$ may not be $C^1$ everywhere
(in fact, $W$ may not even be $C^0$), as the flow $\psi$ is at least $C^1$
the above expression for
$\d{\psi_{(t,p)}(\hat{y}(p))}{p}$ is well-defined and continuous.
As $A_s(p)$ is Hurwitz, and hence the eigenvalues $\Lambda(p)$ are stable,
$\lim_{t \to \infty} \d{\psi_{(t,p)}(\hat{y}(p))}{p} = 0$.

We have that $\phi_{(t+t',p)}(y(p)) = h^{-1} \circ \psi_{(t,p)}(\hat{y}(p))$.
Let $x(t,p) = \psi_{(t,p)}(\hat{y}(p))$.
Then this implies that
$\d{\phi_{(t+t',p)}(y(p))}{p} = \pd{h^{-1}(x(t,p),p)}{p}
+ \pd{h^{-1}(x(t,p),p)}{x}\d{\psi_{(t,p)}(\hat{y}(p))}{p}$.

As $\lim_{t \to \infty} \psi(t,\hat{p}) = 0$ and
$\lim_{t \to \infty} \d{\psi_{(t,\hat{p})}(\hat{y}(\hat{p}))}{p} = 0$,
$\lim_{t \to \infty} \d{\phi_{(t+t',\hat{p})}(y(\hat{p}))}{p} =
\pd{h^{-1}(0,\hat{p})}{p} = \d{X^s(\hat{p})}{p}$, which is finite since $X^s$
is (at least) $C^1$ in $p$, so $\lim_{t \to \infty} \mathcal{F}(t'+t,\hat{p})$
exists and is finite.
As $[0,t']$ is compact, $\mathcal{F}(t,\hat{p})$ achieves a finite maximum
$c_m$ over $t \in [0,t']$.
Thus, either $F(\hat{p}) = c_m$ or
$F(\hat{p}) = \sup_{t \geq 0} \mathcal{F}(t'+t,\hat{p})$ which is finite,
so $F(\hat{p})$ is finite.

Let $\epsilon > 0$.
As $h$ is a diffeomorphism, the derivaties of $h^{-1}$ are continuous.
Therefore, there exist $\delta,\delta' > 0$ such
that $||x||_1 < \delta$ and $p \in B_{\delta'}(\hat{p})$ implies that
$\left|\left|\pd{h^{-1}(x,p)}{p} - \pd{h^{-1}(0,\hat{p})}{p}\right|\right|_1 <
\f{\epsilon}{4}$.
As $\lim_{t \to \infty} \psi(t,\hat{p}) = 0$, there exists $T > 0$ such that
$t \geq T$ implies $||\psi(t,\hat{p})||_1 < \delta$.
By continuity of $\psi(T,\hat{p})$, and since $||\psi(t,p)||_1$ is monotonically
decreasing in $t$ for any $p \in J'$, shrinking $\delta'$ if necessary implies
that for $p \in B_{\delta'}(\hat{p})$ and $t \geq T$, $||\psi(t,p)||_1 < \delta$.
As $\pd{h^{-1}(x,p)}{x}$ is continuous and
$B := \overline{B}_\delta(0) \times \overline{B}_{\delta'}(\hat{p})$ is compact,
there exists $c_m > 0$ finite such that the maximum of
$\left|\left|\pd{h^{-1}(x,p)}{x}\right|\right|_1$ over $B$ is equal to $c_m$.
Furthermore, as $\lim_{t \to \infty} \d{\psi(t,p)}{p} = 0$ and $A_s(p)$ and its
eigenvalues are continuous,
increasing $T$ and shrinking $\delta'$ if necessary implies that
for $t \geq T$ and $p \in B_{\delta'}(\hat{p})$,
$\left|\left|\d{\psi(t,p)}{p}\right|\right|_1 < \f{\epsilon}{4c_mn^3}$ where $n$ is
the dimension of $x$.
Thus, for any $t \geq T$ and $p \in B_{\delta'}(\hat{p})$,\footnote{It is straightforward to verify that for matrices $A$ and $B$,
    $||AB||_1 \leq n^3||A||_1||B||_1$ since $|(AB)_{ij}| \leq ||A||_1||B||_1.$}
\begin{align*}
  &\left|\left|\d{\phi(t,p)}{p} - \d{\phi(t,\hat{p})}{p}\right|\right|_1 \\
  &\leq
  \left|\left|\pd{h^{-1}(\psi(t,p),p)}{p}
  - \pd{h^{-1}(\psi(t,\hat{p}),\hat{p})}{p}
  \right| \right|_1 \\
  &+ \left|\left|\pd{h^{-1}(\psi(t,p),p)}{x}\d{\psi(t,p)}{p}
 - \pd{h^{-1}(\psi(t,\hat{p}),\hat{p})}{x}\d{\psi(t,\hat{p})}{p}\right|\right|_1 \\
 &\leq
  \left|\left|\pd{h^{-1}(\psi(t,p),p)}{p}
  - \pd{h^{-1}(0,\hat{p})}{p}
  \right| \right|_1 \\
  &+ 
  \left|\left|\pd{h^{-1}(0,\hat{p})}{p}
  - \pd{h^{-1}(\psi(t,\hat{p}),\hat{p})}{p}
  \right| \right|_1
  \\
  &+\left|\left|\pd{h^{-1}(\psi(t,p),p)}{x}\d{\psi(t,p)}{p}
  \right|\right|_1 \\
  &+ \left|\left|\pd{h^{-1}(\psi(t,\hat{p}),\hat{p})}{x}\d{\psi(t,\hat{p})}{p}
  \right|\right|_1 \\
  &
  \leq
  \f{\epsilon}{4} + \f{\epsilon}{4}
  + n^3\left|\left|\pd{h^{-1}(\psi(t,p))}{x}\right|\right|_1
  \left|\left|\d{\psi(t,p)}{p}\right|\right|_1 \\
  &+ n^3\left|\left|\pd{h^{-1}(\psi(t,\hat{p}))}{x}\right|\right|_1
  \left|\left|\d{\psi(t,\hat{p})}{p}\right|\right|_1 
  \leq \f{\epsilon}{2} + \f{\epsilon}{4} + \f{\epsilon}{4}
  = \epsilon.
\end{align*}
Furthermore, as $[0,T]$ is compact and $\d{\phi(t,p)}{p}$ is continuous,
shrinking $\delta'$ further if necessary implies that for $t \in [0,T]$
and $p \in B_{\delta'}(\hat{p})$,
$\left|\left|\d{\phi(t,p)}{p} - \d{\phi(t,\hat{p})}{p}\right|\right|_1
\leq \epsilon$.
Thus, for any $p \in B_{\delta'}(\hat{p})$ and $t \geq 0$,
\begin{align*}
  \left|\left|\d{\phi(t,p)}{p}\right|\right|_1 &\leq
  \left|\left|\d{\phi(t,p)}{p}-\d{\phi(t,\hat{p})}{p}\right|\right|_1
  + \left|\left|\d{\phi(t,\hat{p})}{p}\right|\right|_1 \\
  &\leq \left|\left|\d{\phi(t,\hat{p})}{p}\right|\right|_1 + \epsilon.
\end{align*}
Taking the supremum over time implies
\begin{align*}
  F(p) &= \sup_{t \geq 0} \left|\left|\d{\phi(t,p)}{p}\right|\right|_1
  \leq \sup_{t \geq 0} \left|\left|\d{\phi(t,\hat{p})}{p}\right|\right|_1 +
  \epsilon
  = F(\hat{p}) + \epsilon.
\end{align*}
By an analogous argument, reversing the roles of $p$ and $\hat{p}$ we have
$F(\hat{p}) \leq F(p) + \epsilon$.
Thus, $p \in B_{\delta'}(\hat{p})$ implies
$|F(p)-F(\hat{p})| \leq \epsilon$, so $F$ is continuous at $\hat{p}$.
Therefore, as $\hat{p} \in R$ was arbitrary, $F$ is positive, finite, and
continuous over $R$.
\end{proof}

\begin{lemma}
  \label{lem:Rmanifold}
$\partial R$ consists of a finite union of $C^3$ embedded submanifolds.
\end{lemma}

\begin{proof}[Proof of Lemma~\ref{lem:Rmanifold}]
  By Assumption~\ref{as:yj}, $y_J$ is transverse to $W^s(X^i_J)$ for all
  $i \in I$.
  Therefore, $M_i = y_J \cap W^s(X^i_J)$ is a $C^3$ embedded submanifold
  with codimension equal to the sum of the codimensions of $y_J$ and
  $W^s(X^i_J)$ in $M \times J$.
  In particular, the dimension of $M_i$ is less than the dimension of $y_J$,
  which is equal to the dimension of $J$.
  As $y_J = \text{graph }y$, $y:J \to y_J$ is a $C^3$ diffeomorphism onto its
  image
  with inverse $y^{-1} = \pi_J|_{y_J}$ the projection that sends $(x,p) \to p$.
  Thus, $\pi_J|_{y_J}$ is a $C^3$ diffeomorphism onto its image,
  so since $M_i \subset y_J$, its restriction $\pi_J|_{M_i}$ is a $C^3$ embedding
  into $J$.
  Thus, $\pi_J(M_i) = \pi_J|_{M_i}(M_i)$ is a $C^3$ embedded submanifold in $J$
  for each $i \in I$.
  However, by Proposition~\ref{thm:multi} we have
  $\partial R = \pi_J\left(y_J \cap \partial W^s(X^s_J)\right) =
  \pi_J\left(y_J \cap \cup_{i \in I} W^s(X^i_J)\right) = \cup_{i \in I} \pi_J(M_i)$
  so $\partial R$ consists of a finite union of $C^3$ embedded submanifolds
  in $J$.
  Moreover, for any $p^* \in \partial R$ there exists a unique critical element
  $X^{i^*}(p^*)$ such that $y(p^*) \in W^s(X^{i^*}(p^*))$.
  Then we have that $p^* \in M_{i^*} = \pi_J(y_J \cap W^s(X^{i^*}_J))$.
\end{proof}

\begin{lemma}
  \label{lem:onedim}
For any $p^* \in \partial R$, by Proposition~\ref{thm:multi}
$y(p^*) \in W^s(X^*(p^*))$
for some critical element $X^*$.
Then for $p^* \in \partial R$ almost everywhere, $W^s(X^*(p^*))$ has
codimension one.
\end{lemma}

\begin{proof}[Proof of Lemma~\ref{lem:onedim}]
By Lemma~\ref{lem:Rmanifold} and its proof,
$\partial R = \cup_{i \in I} M_i$, where each
$M_i$ is a $C^3$ embedded submanifold in $J$ with codimension equal to the
codimension of $W^s(X^i(p^*))$ for each $p^* \in M_i$.
Note that no $M_i$ can have codimension zero, since this would contradict
that $M_i \subset \partial R$.
As $R$ is open and not dense in $J$,
by \cite[Corollary 2]{Hu48} there must
exist at least one $i^*$ such that $M_{i^*}$ has codimension one.
Therefore, $\partial R$ has positive Lebesgue measure in $(\text{dim } J - 1)$
dimensions.
Furthermore, for each $M_i$ which corresponds to $W^s(X^i(p^*))$ for
$p^* \in M_i$ with codimension greater than one, that $M_i$ also has
codimension greater than one and, therefore,
zero Lebesgue measure in $(\text{dim } J - 1)$ dimensions.
Hence, the set of $p^* \in \partial R$ such that $W^s(X^*(p^*))$ has
codimension greater than one is equal to the finite union of the $M_i$ which
each have codimension greater than one, and therefore has
zero Lebesgue measure in $(\text{dim } J - 1)$ dimensions.
Thus, for $p^* \in \partial R$ almost everywhere, $W^s(X^*(p^*))$ has
codimension one.
\end{proof}

\begin{figure}[t!]
  \centering
  \begin{subfigure}[t]{0.58\linewidth}
    \centering
    \includegraphics[width=\linewidth]{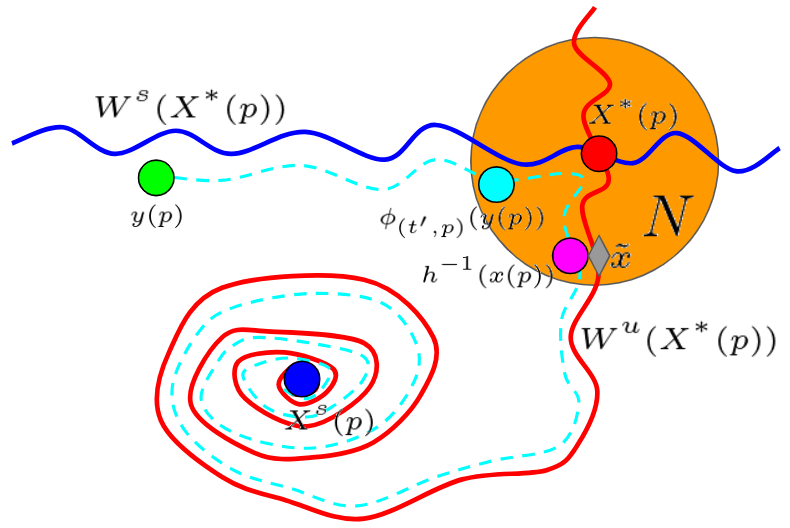}
    \caption{The system trajectory (cyan dashed line) for parameter values
    near the recovery boundary.}
  \end{subfigure}%
  \hfill
 \begin{subfigure}[t]{0.4\linewidth}
   \centering
   \includegraphics[width=\linewidth]{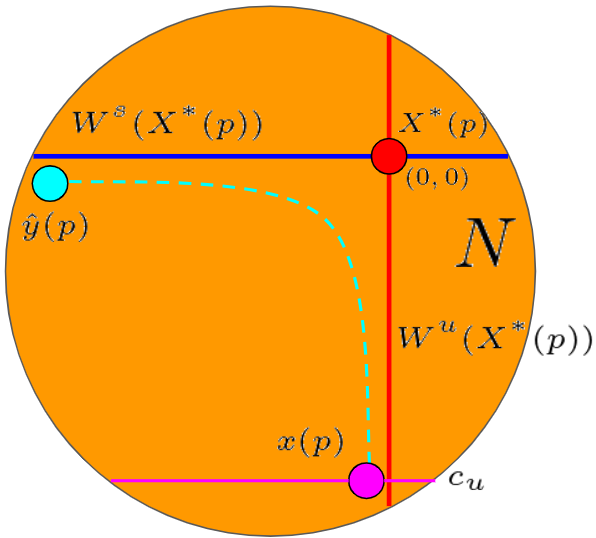}
   \caption{The system trajectory (cyan dashed line) in the smooth local
   coordinates $N$.}
 \end{subfigure}%
 \caption{Local coordinates for the proof of Theorem~\ref{thm:Gtot}.}
 \label{fig:proof}
\end{figure}

\begin{proof}[Proof of Theorem~\ref{thm:Gtot}]
Let $p^* \in \partial R$.
Then $y(p^*) \in W^s(X^*(p^*))$ for some critical element $X^*(p^*)$.
Assume that $W^s(X^*(p^*))$ has codimension one, which is true almost always
by Lemma~\ref{lem:onedim}.
By Assumption~\ref{as:oneD}, $X^*$ is an equilibrium point.
By Assumption~\ref{as:res} and \cite[Theorem 7]{Se85},
there exists a neighborhood $N$ of $X^*(p^*)$, a neighborhood $J'$ of
$p^*$ in $J$, and a $C^3$ conjugacy $h:N \times J' \to \mathbb{R}^n$
known as a smooth linearization
such that for $p \in J'$, $X^*(p) \in N$ and the vector field $V_p$ is
conjugate by $h$ to a linear vector field in $\mathbb{R}^n$.
Shrinking $N$ and $J'$ if necessary, there exist local coordinates on $N$ as
constructed in \cite[p. 81]{Pa82} which straighten out the
stable and unstable manifolds of a hyperbolic equilibrium point so that they
are orthogonal linear subspaces of the local coordinates.
With an abuse of notation, let $h$ denote the composition of the smooth
linearization $h$ with these local coordinates.
Figure~\ref{fig:proof} illustrates the coordinate systems and notation used
in this proof.
Under the local conjugacy $h$, $V_p$ becomes linear and decomposable into
orthogonal eigenspaces representing the (straightened out) stable and unstable
eigenspaces of the linearization of $V_p$ at $X^*(p)$.
In particular, in these coordinates $V_p$ has the 
form
  $V_p(x)  = \begin{bmatrix}
    A_u(p) & 0 \\
    0 & A_s(p)
  \end{bmatrix}$
where $A_u(p)$ and $A_s(p)$ are $C^3$ matrices such that
$e^{A_u(p)} > 1$ (since by Assumption~\ref{as:oneD} the unstable manifold
of $X^*(p)$ is one dimensional) and $||e^{A_s(p)}|| < 1$ for all $p \in J'$.
Let $\pi_u$ and $\pi_s$ denote projection onto the unstable and stable
manifolds (which are also eigenspaces) of the equilibrium point in these
coordinates.
As $y(p^*) \in W^s(X^*(p^*))$ and $X^*(p^*) \in N$, there exists
$t' > 0$ such that $\phi_{(t',p^*)}(y(p^*)) \in N$.
Let $\hat{y}(p) = h \circ \phi_{(t',p)}(y(p))$.
By continuity of the flow and $y(p)$, shrinking $J'$ if necessary
implies that for $p \in J'$, $\hat{y}(p)$ is well-defined and $C^3$.
Let $x_u(p) = \pi_u\hat{y}(p)$ and $x_s(p) = \pi_s\hat{y}(p)$ for $p \in J'$.
Then the dynamics in these coordinates are given by (where we
let $\psi$ denote the flow in this coordinate system) 
$\psi_{(t,p)}(\hat{y}(p)) = \begin{bmatrix}
    e^{A_u(p)t} x_u(p) \\
    e^{A_s(p)t} x_s(p)
    \end{bmatrix}$.
Assume that $A_s(p^*)$ has distinct eigenvalues, which is generically true.
As this is an open condition, shrinking $J$ if necessary implies that $A_s(p)$
is diagonalizable for all $p \in J$.
Thus, we can write $A_s(p) = W(p) \Lambda(p) W(p)^{-1}$ where $\Lambda(p)$
is diagonal and $W(p)$ is invertible.  Note that by the implicit function
theorem, and since the eigenvalues are distinct, $\Lambda(p)$ and $V(p)$ are
(at least) $C^3$ over $J$.
Define $\d{W(p)}{p} = \begin{bmatrix} \pd{W(p)}{p_1} & \hdots &
  \pd{W(p)}{p_{\text{dim } J}} \end{bmatrix}$
and $\d{\Lambda(p)}{p} = \begin{bmatrix} \pd{\Lambda(p)}{p_1} & \hdots &
  \pd{\Lambda(p)}{p_{\text{dim } J}} \end{bmatrix}$.
Then the derivative of the flow with respect to $p$ is given by the following,
where we omit the dependence on $p$ for brevity
\begin{align}
  &\d{\psi_{(t,p)}(\hat{y}(p))}{p} = \begin{bmatrix}
      t e^{A_u t} \d{A_u}{p} x_u + e^{A_u t} \d{x_u}{p} \\
  \d{W}{p}\left(I \otimes W^{-1}e^{A_s t}x_s\right) \end{bmatrix} \nonumber \\
  &+ \begin{bmatrix}
    0 \\
     e^{A_s t}\left(\left(tW\d{\Lambda}{p}- \d{W}{p}\right)
    \left(I \otimes W^{-1}x_s\right) + \d{x_s}{p} \right)
     \end{bmatrix}.
    \label{eq:dphi}
\end{align}

As $\phi_{(t',p)}(y(p^*)) \in W^s(X^*(p^*)) \cap N$ and
$h(W^s(X^*(p^*)) \cap N) = \{x:~\pi_ux = 0\}$, $x_u(p^*) = 0$.
Note that, by Assumption~\ref{as:yj}, $\d{x_u(p^*)}{p} \neq 0$ since otherwise
$y(J)$ would not be transversal to $W^s(X^s(p^*))$ at $p^*$.
Then evaluating \eqref{eq:dphi} at $p^*$ implies
$\pi_u \d{\psi_{(t,p^*)}(\hat{y}(p^*))}{p} = e^{A_u(p^*) t} \d{x_u(p^*)}{p}$
since $x_u(p^*) = 0$. Thus, we have
$\lim_{t \to \infty} \pi_u \d{\phi_{(t,p^*)}(\hat{y}(p^*))}{p} =
\infty \d{x_u(p^*)}{p}$ since
$\d{x_u(p^*)}{p} \neq 0$ and $A_u(p^*)$ is an unstable eigenvalue.
Assume that $\d{x_u(p^*)}{p}$ is nonzero in each entry, which is generically
true.
The flow in the original coordinates is given by
$\phi_{(t'+t,p)}(y(p)) = h(\cdot,p)^{-1} \circ \psi_{(t,p)}(\hat{y}(p))$.
Taking a derivative with respect to $p$ yields
  $\d{\phi_{(t'+t,p)}(y(p))}{p}
  =
  \pd{h^{-1}(\psi_{(t,p)}(\hat{y}(p)),p)}{p}
  + \pd{h^{-1}(\psi_{(t,p)}(\hat{y}(p)),p)}{x} \d{\psi_{(t,p)}(\hat{y}(p))}{p}$.
As $h$ is a diffeomorphism, $\pd{h^{-1}}{x}$ is full rank everywhere.
Thus, by \eqref{eq:dphi}, at $p = p^*$ we have
$\sup_{t \geq 0}  \left|\left|\d{\phi_{(t'+t,p)}(y(p))}{p}\right|\right|_1
= \infty$, which implies that $G(p^*) = 0$.
For any $p \in R \cap J'$, $x_u(p) \neq 0$ since otherwise
$\hat{y}(p)$ would lie in the stable manifold of $X^*(p)$, which it cannot
by definition of $R$.
Without loss of generality, suppose $x_u(p) > 0$, and fix a 
constant $c_u > 0$.
For $p \in R \cap J'$,
let $t(p)$ denote the time
that
$\pi_u \psi_{(t(p),p)}(\hat{y}(p)) = c_u$.
Then solving for $t(p)$ yields
$t(p) = \f{1}{A_u(p)}\log \f{c_u}{x_u(p)}$.

As $p^* \in \partial R$, by Lemma~\ref{lem:Rmanifold} and its proof there
exists a $C^3$ embedded submanifold $\hat{M}$ in $J$ with codimension one such
that $p^* \in \hat{M}$.
As $\hat{M}$ is an embedded submanifold, there exists a local slice chart
centered at $p^*$: i.e. local coordinates in $J$
such that the intersection of $\hat{M}$ with this coordinate neighborhood is
equal to $\mathbb{R}^{\text{dim }J - 1} \times \{0\}$.
Shrink $J'$ if necessary so that $J'$ is contained in the neighborhood of
these slice chart coordinates.
As $\hat{M} \subset \partial R$, without loss of generality we may
assume that in this slice chart any point $(p,s)$ with $s < 0$ lies in $R$,
and any point $(p,s)$ with $s > 0$ lies outside of $\overline{R}$.
In this slice chart, abuse notation to represent $p^*$ as $(p^*,0)$ and
let $\gamma:[0,1] \to \overline{R}$ be represented 
by
$\gamma(s) = (p^*,s-1)$ for all $s \in [0,1]$.
Then $\gamma(s) \in R$ for all $s \in [0,1)$ and $\gamma(1) = p^*$.
Also, $\gamma'(1) = (0,1) \neq 0$, so $\gamma$ is transverse to
$\hat{M}$, and thus to $\partial R$, at $p^*$.

For $p \in \gamma$, say $p = \gamma(s)$, define
$\hat{y}'(\gamma'(s)) = d(\hat{y} \circ \gamma)_s$.
Define $\hat{y}''(\gamma''(s))$ and $\hat{y}'''(\gamma'''(s))$ analogously
for the second and third derivatives, which exist and are continuous since
$y$ is $C^3$.
For the remainder of the proof, let $x_u(s) = \pi_u \hat{y}(\gamma(s))$,
$x_s(s) = \pi_s \hat{y}(\gamma(s))$,
and define $x_u'(s)$, $x_s'(s)$, $x_u''(s)$, $x_s''(s)$, $x_u'''(s)$, and
$x_s'''(s)$ analogously
by composing the $i$th derivative of $\hat{y}$ with
$\pi_u$ or $\pi_s$ for $i \in \{1,2,3\}$, respectively.
By Assumption~\ref{as:yj}, $x_u'(\gamma'(1)) \neq 0$ since otherwise
$y(J)$ would not be transversal to $W^s(X^s(p^*))$ at $p^*$.
As $y(J)$ is transverse to $W^s(X^*(J))$ by Assumption~\ref{as:yj} and $h$ is
$C^3$,
$\hat{y}(J)$ is transverse to $W^s(X^*(J))$ in the 
coordinates of $h$.
This implies that $\pi_u d(\hat{y})_{p^*}$ is full rank (i.e. nonzero).
So, since $\gamma'(1) \neq 0$, 
$x_u'(\gamma'(1)) = \pi_u d(\hat{y})_{p^*}(\gamma'(1)) \neq 0$.
For brevity 
we often omit the 
dependence
on $s$.
For $p \in \gamma \cap R \cap J'$, say $p = \gamma(s)$, define
\begin{align}
  &\Phi(p) := \d{\psi_{(t,p)}(\hat{y}(p))}{p}|_{t=t(p)}
  = \begin{bmatrix} \pi_u \Phi(p) \\ \pi_s \Phi(p) \end{bmatrix}
  \nonumber \\
  &\Phi(s) := \Phi(\gamma(s))
  = \begin{bmatrix}
    \f{c_u}{A_u} \d{A_u}{p} \log \f{c_u}{x_u} + \d{x_u}{p} \f{c_u}{x_u} \\
    \d{W}{p}\left(I \otimes W^{-1}e^{\f{A_s}{A_u} \log \f{c_u}{x_u}}x_s\right)
  \end{bmatrix} \nonumber \\
  &+ \begin{bmatrix}
    0 \\
    e^{\f{A_s}{A_u} \log \f{c_u}{x_u}} \left(
    \left(\f{\log \f{c_u}{x_u}}{A_u}  W\d{\Lambda}{p}
    - \d{W}{p}\right)
    \left(I \otimes W^{-1}x_s\right) + \d{x_s}{p} \right)
  \end{bmatrix} \nonumber \\
    &=: \begin{bmatrix}
    B_1(s) \log \f{c_u}{x_u} + B_2(s) \f{1}{x_u} \\
    B_4(s)\left(I \otimes W^{-1}e^{B_3(s) A_s \log \f{c_u}{x_u}}x_s\right)
  \end{bmatrix} \nonumber \\
  &+ \begin{bmatrix}
    0 \\
    e^{B_3(s) A_s\log \f{c_u}{x_u}} \left(\log \f{c_u}{x_u} B_5(s) + B_6(s) \right)
    \end{bmatrix}
    \label{eq:phis}    
\end{align}
where it is straightforward to verify that
$B_1(s)$, $B_2(s)$, $B_3(s)$, $B_4(s)$, $B_5(s)$, and $B_6(s)$ are $C^2$.
Note $B_3(s)A_s(s) = \f{A_s(s)}{A_u(s)}$ is Hurwitz since $A_s(s)$ is
Hurwitz and $A_u(s) > 0$.
For the following, we omit the explicit dependence on $s$
and indicate derivatives with respect to $s$ using primes.
We compute
\begin{align}
  & \pi_s \Phi'(s)
  = B_4'\left(I \otimes W^{-1}e^{B_3 A_s \log \f{c_u}{x_u}}x_s\right) \nonumber \\
  &+ B_4 \left(I \otimes W^{-1}e^{B_3 A_s \log \f{c_u}{x_u}}
  \left(x_s' + \left(-W' \vphantom{\log \f{c_u}{x_u}} \right. \right. \right.
  \nonumber \\
  & \left. \left. \left.
  + W\left(\log \f{c_u}{x_u}(\Lambda'B_3 + \Lambda B_3')
  -\Lambda B_3 \f{x_u'}{x_u}\right)
  \right)W^{-1}x_s \right)\right) \nonumber \\
  &+ \left(W'W^{-1}e^{B_3 A_s \log \f{c_u}{x_u}} + e^{B_3 A_s \log \f{c_u}{x_u}}
  \left(
  W\left(
  -\Lambda B_3 \f{x_u'}{x_u} \right. \right. \right. \nonumber \\
  & \left. \left. \left.
  +\log \f{c_u}{x_u}(\Lambda'B_3 + \Lambda B_3')
  \right)
  - W'\right)W^{-1}
  \right)\left(\log \f{c_u}{x_u} B_5 + B_6 \right) \nonumber \\
  &
  +  e^{B_3 A_s\log \f{c_u}{x_u}} \left(\log \f{c_u}{x_u} B_5'
  - B_5 \f{x_u'}{x_u} + B_6' \right) \nonumber \\
  & =: B_4'\left(I \otimes W^{-1}e^{B_3 A_s \log \f{c_u}{x_u}}x_s\right) \nonumber \\
  &+ B_4 \left(I \otimes W^{-1}e^{B_3 A_s \log \f{c_u}{x_u}}
  \left(B_7 \log \f{c_u}{x_u} + B_8 \f{x_u'}{x_u} + B_9\right)\right) \nonumber
  \\
  &+ \left(e^{B_3 A_s \log \f{c_u}{x_u}}
   \left(B_{11} \log \f{c_u}{x_u} + B_{12} \f{x_u'}{x_u} + B_{13}\right) \right.
   \nonumber \\
   & \left. + B_{10}e^{B_3 A_s \log \f{c_u}{x_u}} \vphantom{\log \f{c_u}{x_u}}\right)
   \left(\log \f{c_u}{x_u} B_5 + B_6 \right) \nonumber \\
   &+  e^{B_3 A_s\log \f{c_u}{x_u}} \left(\log \f{c_u}{x_u} B_5'
   - B_5 \f{x_u'}{x_u} + B_6' \right)
   \label{eq:phiss}
\end{align}
where it is straightforward to verify that
$B_7(s)$, $B_8(s)$, $B_9(s)$, $B_{10}(s)$, $B_{11}(s)$, $B_{12}(s)$, and
$B_{13}(s)$ are $C^1$.  We compute
\begin{align}
  & \pi_s \Phi''(s) =
  B_4''\left(I \otimes W^{-1}e^{B_3 A_s \log \f{c_u}{x_u}}x_s\right) \nonumber \\
  & + 2B_4' \left(I \otimes W^{-1}e^{B_3 A_s \log \f{c_u}{x_u}}
  \left(B_7 \log \f{c_u}{x_u} + B_8 \f{x_u'}{x_u} + B_9\right)\right) \nonumber
  \\
  &+B_4 \left(I \otimes
  W^{-1}e^{B_3 A_s \log \f{c_u}{x_u}}
  \left(W\left((\Lambda'B_3 + \Lambda B_3')\log \f{c_u}{x_u}
  \right. \right. \right. \nonumber \\
  & \left. \left. \left. - \Lambda B_3
  \f{x'}{x}\right) - W'\right)W^{-1}
\left(B_7 \log \f{c_u}{x_u} + B_8 \f{x_u'}{x_u} + B_9\right)\right) \nonumber \\
&+B_4 \left(I \otimes W^{-1}e^{B_3 A_s \log \f{c_u}{x_u}}
\left(B_7' \log \f{c_u}{x_u} - B_7 \f{x_u'}{x_u}
+ B_8' \f{x_u'}{x_u} \right. \right. \nonumber \\
& \left. \left. + B_8 \f{x_u''x_u-(x_u')^2}{x_u^2}
+ B_9'\right) \right) \nonumber \\ 
&+ \left(\left(
e^{B_3 A_s \log \f{c_u}{x_u}}
\left(B_{11} \log \f{c_u}{x_u} + B_{12} \f{x_u'}{x_u} + B_{13}
\right) \right. \right. \nonumber \\
&+ \left. B_{10}e^{B_3 A_s \log \f{c_u}{x_u}} \right)
\left(B_{11} \log \f{c_u}{x_u} + B_{12} \f{x_u'}{x_u} + B_{13}\right) \\
& 
+ e^{B_3 A_s \log \f{c_u}{x_u}}\left(B_{11}'\log \f{c_u}{x_u} - B_{11} \f{x_u'}{x_u}
+ B_{12}' \f{x_u'}{x_u} \right. \nonumber \\
& \left. + B_{12} \f{x_u''x_u - (x_u')^2}{x_u^2} + B_{13}'\right)
+ B_{10}'e^{B_3 A_s \log \f{c_u}{x_u}} \nonumber \\
&  + B_{10}\left(
B_{10}e^{B_3 A_s \log \f{c_u}{x_u}} + e^{B_3 A_s \log \f{c_u}{x_u}}
\left(B_{11} \log \f{c_u}{x_u}  \right. \right. \nonumber \\
& \left. \left. \left. + B_{12} \f{x_u'}{x_u} + B_{13}
\vphantom{\log \f{c_u}{x_u}} \right) \right) \right)
\left(\log \f{c_u}{x_u} B_5 + B_6 \right) \nonumber \\
&+ 2\left(e^{B_3 A_s \log \f{c_u}{x_u}}
\left(B_{11} \log \f{c_u}{x_u} + B_{12} \f{x_u'}{x_u} + B_{13}\right) \right.
\nonumber \\
& \left. + B_{10}e^{B_3 A_s \log \f{c_u}{x_u}} \right)
\left(\log \f{c_u}{x_u} B_5' - B_5 \f{x_u'}{x_u} + B_6' \right) \nonumber \\
&+ e^{B_3 A_s \log \f{c_u}{x_u}}
\left(\log \f{c_u}{x_u} B_5'' - B_5' \f{x_u'}{x_u}
- B_5' \f{x_u'}{x_u} \right. \nonumber \\
& \left. - B_5 \f{x_u''x_u-(x_u')^2}{x_u^2} + B_6'' \right)
\label{eq:phisss}
\end{align}
We compute
\begin{align}
  &\pi_u \Phi'(s) =
  B_1' \log \f{c_u}{x_u} + (B_2' - B_1 x_u') \f{1}{x_u} - B_2 x_u' \f{1}{x_u^2}
  \label{eq:phisu} \\
  &\pi_u \Phi''(s)
  = B_1'' \log \f{c_u}{x_u} + (B_2'' -2 B_1'x_u' - B_1 x_u'') \f{1}{x_u}
  \nonumber \\
  &+ (B_1(x_u')^2-2 B_2'x_u'-B_2x_u'') \f{1}{x_u^2}
  + 2 B_2 (x_u')^2 \f{1}{x_u^3}
  \label{eq:phisuu}
\end{align}

Let $x(p) = \psi_{(t(p),p)}(\hat{y}(p))$ for $p \in J' \cap R$.
We can write
$\phi_{(t,p)}(x(p)) = \phi_{(t,p)} \circ h(\cdot,p)^{-1} \circ
  \psi_{(t(p),p)}(\hat{y}(p))$.
Taking a derivative implies
\begin{align}
\d{\phi_{(t,p)}(x(p))}{p} &= \pd{\phi_{(t,p)}(x(p))}{p} +\pd{\phi_{(t,p)}(x(p))}{x}
\nonumber \\
&*\left(\pd{h^{-1}(x(p),p)}{p}+ \pd{h^{-1}(x(p),p)}{x}\Phi(p) \right).
\label{eq:Phi}
\end{align}
For $t \geq 0$, we define the function
\begin{align*}
  f(t,p) &= x_u(p)
  \left|\left|\pd{\phi_{(t,p)}(x(p))}{p} +\pd{\phi_{(t,p)}(x(p))}{x}
  \right.\right. \\
  &\left.\left.*\left(\pd{h^{-1}(x(p),p)}{p}+ \pd{h^{-1}(x(p),p)}{x}
  \f{c_u}{x_u(p)} \d{x_u(p)}{p} \right)
  \right|\right|_1.
\end{align*}
Let
$\tilde{x} = h^{-1}\left(\begin{bmatrix} c_u & 0 \end{bmatrix}^\intercal\right)$.
Note $\lim_{p \to p^*} x(p) = \tilde{x}$.
As $x_u(p^*) = 0$, 
$f(t,p^*) = c_u\left|\left|\pd{\phi_{(t,p^*)}(\tilde{x})}{x}
\pd{h^{-1}(\tilde{x},p^*)}{x} \d{x_u(p^*)}{p}\right|\right|_1$
is finite and well-defined.
As $\lim_{t \to \infty} \phi_{(t,p)}(x) = X^s(p)$ for any $p \in R$ and
$x \in W^s(X^s(p))$,
$\lim_{t \to \infty} \pd{\phi_{(t,p)}(x)}{x} = 0$.
Thus, since $\pd{\phi_{(t,p)}(x)}{x}$, $\pd{h^{-1}(x,p)}{x}$,
and $\d{x_u(p)}{p}$ are all full rank for $p$ near $p^*$,
$f(t,p)$ is strictly positive for
all $t \geq 0$ and zero in the limit as $t \to \infty$, so it must have a
finite supremum over $t \geq 0$.  As $f$ is continuous, this supremum must
be a maximum that is attained in finite time.
By Assumption~\ref{as:unique}, there exists a unique time $\tilde{t}$ at which
the maximum of $f(t,p^*)$ over $t \geq 0$ is attained.
As this is a local maximum, $\pd{f(\tilde{t},p^*)}{t} = 0$ and
$\f{\partial^2 f(\tilde{t},p^*)}{\partial t^2} < 0$.
Thus, by the implicit function theorem there exist neighborhoods $U$ of $p^*$
and $T$ of $\tilde{t}$, and a $C^2$ function $\hat{t}: U \to T$ such
that $\hat{t}(p^*) = \tilde{t}$ and $p \in U$ implies
$\pd{f(\hat{t}(p),p)}{t} = 0$.
As $f(t,p)$ attains a maximum over $t \geq 0$ in finite time for any $p$
near $p^*$, and since $f$ is $C^2$, shrinking $U$ if necessary implies that
for each $p \in U$, $\hat{t}(p)$ will be the unique global maximum of
$f(t,p)$ over $t \geq 0$.
Shrink $J'$ if necessary so that $J' \subset U$.

As $x_u(p)$ is constant in time, note that
$\argmax_{t \geq 0} \f{1}{x_u(p)} f(t,p) = \argmax_{t \geq 0} f(t,p) = \hat{t}(p)$.
In the limit as $p \to p^*$, \eqref{eq:phis} implies that
$\Phi \to \begin{bmatrix} \pi_u \Phi & 0 \end{bmatrix}^\intercal$
and that $\pi_u \Phi \to \f{c_u}{x_u(p)} \d{x_u(p)}{p}$.
Thus, in the limit as $p \to p^*$, by \eqref{eq:Phi}
$\left|\left|\d{\phi_{(t,p)}(x(p))}{p}\right|\right|_1 \to \f{1}{x_u(p)}f(t,p)$,
so it attains a unique global maximum over $t \geq 0$ at $\hat{t}(p)$.
Then for $t \geq 0$, we define the function
  $H(t,p) = \left|\left|\d{\phi_{(t,p)}(x(p))}{p}\right|\right|_1^{-1}$
and note that in the limit as $p \to p^*$,
  $G(p) = \inf_{t \geq 0} H(t,p) =
  \left(\sup_{t \geq 0} \left|\left|\d{\phi_{(t,p)}(x(p))}{p}\right|\right|_1
  \right)^{-1} 
  =   \left|\left|\d{\phi_{(\hat{t}(p),p)}(x(p))}{p}\right|\right|_1^{-1}
  = H(\hat{t}(p),p)$.
Write $p = \gamma(s)$.
Then we define the functions
\begin{align*}  
  C(s) &= \pd{\phi_{(\hat{t}(p),p)}(x(p))}{p} + \pd{\phi_{(\hat{t}(p),p)}(x(p))}{x}
  \pd{h^{-1}(x(p),p)}{p} \\
  D(s) &= \pd{\phi_{(\hat{t}(p),p)}(x(p))}{x}\pd{h^{-1}(x(p),p)}{x}.
\end{align*}
Then the above implies that for $p = \gamma(s) \in J' \cap R$,
\begin{align}
  G(s) := G(\gamma(s)) = H(\hat{t}(p),p) =
  \left|\left|C(s)+D(s)\Phi(s)\right|\right|_1^{-1}. \label{eq:G}
\end{align}
Define $\text{sign}(M)$ by $\text{sign}(M)_{ij} = \text{sign}(M_{ij})$.
We compute, omitting the explicit dependence on $s$, and letting
primes denote derivatives with respect to $s$,
\begin{align}
  G' &= -G^2 \text{vec}(\text{sign}(C+D\Phi))^\intercal\text{vec}
  (C'+D'\Phi+D\Phi') \label{eq:gp} \\
  G'' &= 2 G^3 \left(\text{vec}(\text{sign}(C+D\Phi))^\intercal\text{vec}
  (C'+D'\Phi+D\Phi')\right)^2 \nonumber \\
  &\mkern-20mu
  - G^2\text{vec}(\text{sign}(C+D\Phi))^\intercal
  \text{vec}(C'' + D''\Phi + 2D'\Phi' + D\Phi'') \label{eq:gpp}
\end{align}
where $G''$ is well-defined only in regions where $\text{sign}(C+D\Phi)$
is locally constant.

Next, we will take the limit as $p \to p^*$ along $\gamma$, which is equivalent
to the limit as $s \to 1$.
Note that each term in \eqref{eq:G}-\eqref{eq:gpp} can be expressed as a signed
sum of the elements in $C+D\Phi$ or its derivatives.
Thus, in the limit as $s \to 1$, these terms will be dominated by the entries
of $C + D\Phi$ (or its derivatives) with the fastest asymptotic growth as
$s \to 1$.
Based on \eqref{eq:phis}-\eqref{eq:phisuu}, in the limit as $s \to 1$ the only
terms which do not converge are the terms involving $\f{1}{x_u}$
because $\lim_{s \to 1} x_u(s) = x_u(p^*) = 0$.
In $\pi_s \Phi$, $\pi_s \Phi'$, and $\pi_s \Phi''$, the fastest growth rate is
$e^{\f{A_s}{A_u} \log \f{c_u}{x_u}} \f{1}{x_u}$, which is strictly slower than
$\f{1}{x_u}$ because $\f{A_s}{A_u}$ is Hurwitz since $A_s$ is Hurwitz
and $A_u > 0$.
In contrast, the fastest growth rates for $\pi_u \Phi$, $\pi_u \Phi'$,
and $\pi_u \Phi''$ are $\f{1}{x_u}$, $\f{1}{x_u^2}$, and $\f{1}{x_u^3}$,
respectively, and the coefficients of these terms are nonzero because they
are a product of $B_2 = \d{x_u}{p}$ and powers of $x_u'$, both of which are
nonzero at $s = 1$ by Assumption~\ref{as:yj} since $y(J)$ is transverse to
$W^s(X^s(p^*))$, and since $\gamma$ is transverse to $\hat{M}$.
Thus, in the limit as $s \to 1$, the entries in $\Phi$, $\Phi'$, and $\Phi''$
with the fastest growth rates lie within $\pi_u \Phi$, $\pi_u\Phi'$, and
$\pi_u\Phi''$, respectively, and dominate all the entries in
$\pi_s \Phi$, $\pi_s\Phi'$, and
$\pi_s\Phi''$, respectively.
Note that $D(1)$ is full rank since the flow $\phi$ and local coordinates $h$
are each diffeomorphisms.
Therefore, based on the terms appearing in \eqref{eq:G}-\eqref{eq:gpp},
for purposes of computing the limits as $s \to 1$ of $G$, $G'$, and $G''$,
it suffices to consider $\begin{bmatrix} \pi_u \Phi \\ 0 \end{bmatrix}$,
$\begin{bmatrix} \pi_u \Phi' \\ 0 \end{bmatrix}$,
and $\begin{bmatrix} \pi_u \Phi'' \\ 0 \end{bmatrix}$
in place of $\Phi$, $\Phi'$, and $\Phi''$, respectively.
Furthermore, this implies
\begin{align*}
  &\lim_{s \to 1} \text{vec}(\text{sign}(C(s) + D(s)\Phi(s))) \\
  &= \lim_{s \to 1} \text{vec}\left(\text{sign}\left(C(1) + D(1)
  \begin{bmatrix} \pi_u\Phi(s) \\ 0 \end{bmatrix} \right)\right) \\
  &= \text{vec}\left(\text{sign}\left(D(1)
  \begin{bmatrix} B_2(1) \\ 0 \end{bmatrix} \right)\right) =: \sigma
\end{align*}
which is nonzero (and thus constant sign) in each entry
since $B_2(1) = c_u\d{x_u(p^*)}{p}$ is nonzero in each entry.
Thus, $G''$ is well-defined for $s$ sufficiently close to $1$.
By \eqref{eq:phis}-\eqref{eq:gpp} 
we have
\begin{align*}
  \lim_{s \to 1} G(s)
  &= 0 = G(1) \\
  \lim_{s \to 1} G'(s) &= x_u'(1) \left|\left|D(1)\begin{bmatrix} B_2(1) \\ 0
  \end{bmatrix} \right|\right|^{-1} \\
  \lim_{s \to 1} G''(s) &=
  \sigma^\intercal
  \text{vec}\left(-2 x_u' D'(1) \begin{bmatrix} B_2(1) \\ 0 \end{bmatrix}
  \right.\\
  & - 2 x_u' D(1) \begin{bmatrix} B_2'(1) \\ 0 \end{bmatrix}
  + 3 (x_u')^2 D(1) \begin{bmatrix} B_1(1) \\ 0 \end{bmatrix} \\
  & \left. + x_u'' D(1) \begin{bmatrix} B_2(1) \\ 0 \end{bmatrix}\right)
  \left|\left|D(1)\begin{bmatrix} B_2(1) \\ 0
  \end{bmatrix} \right|\right|^{-2}
\end{align*}
Furthermore, using L'Hospital's Rule we compute
\begin{align*}
  G'(1) &= \lim_{s \to 1} \f{G(s)-G(1)}{s-1}
  = \lim_{s \to 1} G'(s) \\
  G''(1) &= \lim_{s \to 1} \f{G'(s)-G'(1)}{s-1}
  = \lim_{s \to 1} G''(s).
\end{align*}
Thus, $G$, $G'$, and $G''$ are all continuous at $p^* = \gamma(1)$ along
any smooth curve $\gamma$ from inside $R$ to $\partial R$.
In particular, $G$ is $C^2$ along any such curve.

Next, we extend $G$ to a neighborhood of $p^*$ in $J'$ using the slice
chart coordinates constructed above.
As $G$ is already well-defined for any $p \in J' \cap \overline{R}$,
we define $\tilde{G}$ for $p \not\in \overline{R}$ as follows.
If $p \not\in \overline{R}$, we can abuse notation to write it as
$(p,s)$ for some $s > 0$ in the slice chart coordinates.
Then we define $\tilde{G}((p,s)) = -\tilde{G}((p,-s))$, which is well-defined
since $(p,-s) \in R \cap J'$.
Let $\hat{\gamma}(s) = (p^*,1-s)$.
Then $\tilde{G}(\hat{\gamma}(s)) = -G(\gamma(s))$ for all $s \in [0,1]$,
so the limit as $s \to 1$ of $\tilde{G}(\hat{\gamma}(s))$ and its first and
second derivatives is equal to the same limit of $G(\gamma(s))$ and
its first and second derivatives.
This implies that $\tilde{G}$ and its first and second partial derivatives in
the direction of the $s$ coordinate in the slice chart exist and are continuous.
The other standard coordinate directions for approaching $(p^*,0)$ in the
slice chart are all in the set with $s$ coordinate equal to zero, which is
precisely the linear subspace that lies in $\partial R$, and therefore
$\tilde{G}$ and its first and second derivatives are identically zero over this
subspace.
Thus, $\tilde{G}$ and its first and second partial derivatives exist and are
continuous in all of the standard coordinate directions in the slice chart,
so $\tilde{G}$ is $C^2$ at $p^*$.
Furthermore, by definition of $\tilde{G}$, $\partial R = \{p:\tilde{G}(p) = 0\}$
and $R = \{p:\tilde{G}(p) > 0\}$ since $\tilde{G}(p) < 0$ for all
$p \not\in \overline{R}$.

\end{proof}

\section{Proofs of Theorems~\ref{thm:1dim}-\ref{thm:conv}}\label{sec:proofs2}

\begin{lemma}
  \label{lem:conv}
Suppose $p^* \in J$, $r > 0$, and $k \in (0,1)$ satisfy
  $||\tilde{F}(p)-p^*|| \leq k ||p-p^*||$ 
for all $p \in \overline{B}_r(p^*)$, and one of the following:
\begin{enumerate}
\item[(a)] $R \cap \overline{B}_r(p^*)$ is forward invariant under $F$
  and $\partial R \cap \overline{B}_r(p^*) = \{p^*\}$.
\item[(b)] For some hyperplane $H$,
  $R \cap \overline{B}_r(p^*) \cap H$ is forward invariant under $F$
  and $\partial R \cap \overline{B}_r(p^*) \cap H = \{p^*\}$.
\item[(c)] $R \cap \overline{B}_r(p^*)$ is forward invariant under $F$
  and for each $\hat{p} \in \partial R \cap \overline{B}_r(p^*)$ with
  $\hat{p} \neq p^*$, $\tilde{F}(\hat{p}) \in R$.
\end{enumerate}
Then the sequence $\{p^s\}_{s=1}^\infty$ with $p^1 = p_0$ and generated by
the rule $p^{s+1} = F(p^s)$ converges to $p^*$.
\end{lemma}

\begin{proof}[Proof of Lemma~\ref{lem:conv}]
For cases (a) and (c) let $B = \overline{B}_r(p^*)$, and for case (b) let
$B = \overline{B}_r(p^*) \cap H$.
Recall that for any $p \in R \bigcap B$,
$F(p) = F_{m(p)}(p) = p + \f{1}{2^{m(p)}}\left(\tilde{F}(p)-p\right)$
where $m(p)$ is the smallest integer such that $F_{m(p)}(p) \in R$, which is
finite since $R$ is open.
Furthermore, for any nonnegative integer $m$ let
  $k_m = \left(1-\f{1}{2^m}(1-k)\right)$
and note $k_m \in (0,1)$ since $k \in (0,1)$.
Then for any $p \in B$,
\begin{align*}
  ||F_m(p)-p^*||_P
  &= \left|\left|\left(1-\f{1}{2^m}\right)(p-p^*) +
  \f{\tilde{F}(p)-p^*}{2^m}\right|\right|_P \\
  & \mkern-100mu \leq \left(1-\f{1}{2^m}(1-k)\right) ||p-p^*||_P
  = k_m ||p-p^*||_P.
\end{align*}

As $p^1 \in R \cap B$, by forward invariance $p^s \in R \cap B$ for all
$s \geq 1$.
As $\{p^s\}_{s=1}^\infty \subset R \cap B \subset \overline{R} \cap B$ and
$\overline{R} \cap B$ is compact, this implies that the $\omega$ limit set
of $\{p^s\}_{s=1}^\infty$ is nonempty and contained in
$\overline{R} \cap B$.
So, let $q$ be in the $\omega$ limit set of $\{p^s\}_{s=1}^\infty$.
Then there exists a subsequence $\{p^{s_n}\}_{n=1}^\infty$ such that
$\lim_{n \to \infty} p^{s_n} = q$.
We will show that $q = p^*$.

First suppose $q \in \partial R$.
For cases (a) or (b) this implies that $q = p^*$, so suppose case (c) holds.
Assume towards a contradiction that $q \neq p^*$.
Then since $q \in \partial R \cap B$, $\tilde{F}(q) \in R$.
Thus, since $\tilde{F}$ is continuous and $R$ is open, there exists $N > 0$
sufficiently large such that $n \geq N$ implies that $\tilde{F}(p^{s_n}) \in R$.
So, $n \geq N$ implies $F(p^{s_n}) = \tilde{F}(p^{s_n})$.
Then
$||F(p^{s_n})-p^*||_P = ||\tilde{F}(p^{s_n})-p^*||_P
  \leq k^{n-N}||F(p^{s_N})-p^*||_p$.
As $k \in (0,1)$, taking the limit as $n \to \infty$ implies that
$\lim_{n \to \infty} ||F(p^{s_n})- \nolinebreak p^*||_P = 0$,
which implies that $F(p^{s_n}) \to p^*$.
As $F(p^{s_n}) \to q$, this implies that $q = p^*$, which yields a contradiction.
So, we must have $q = p^*$.

Next suppose $q \in R$.
As $F(q) = F_{m(q)} \in R$, since $F_{m(q)}$ is continuous and $R$ is open,
there exists $N > 0$ sufficiently large such that $n \geq N$ implies that
$F_{m(q)}(p^{s_n}) \in R$.
This implies that for $n \geq N$, $m(p^{s_n}) \leq m(q)$.
Thus, $n \geq N$ implies $F(p^{s_n}) \in \{F_m(p^{s_n})\}_{m=0}^{m(q)}$.
Note that for $m_1 < m_2$, $k_{m_1} < k_{m_2}$.
Hence, for $n \geq N$,
\begin{align*}
  ||F(p^{s_n})-p^*||_P &\leq \max\{||F_m(p^{s_n})-p^*||_P\}_{m=0}^{m(q)} \\
  &\leq \max \{k_m\}_{m=0}^{m(q)} ||F(p^{s_{n-1}})-p^*||_P\\
  &= k_{m(q)} ||F(p^{s_{n-1}})-p^*||_P.
\end{align*}
Applying this iteratively we obtain
$||F(p^{s_n})- p^*||_P \leq \nolinebreak k_{m(q)}^{n-N} ||F(p^{s_N})-p^*||_P$.
As $k_{m(q)} \in (0,1)$, taking the limit as $n \to \infty$ implies that
$\lim_{n \to \infty} ||F(p^{s_n})-p^*||_P = 0$,
which implies that $F(p^{s_n}) \to p^*$.
As $F(p^{s_n}) \to q$, this implies that $q = p^* \in \partial R$, which
contradicts that $q \in R$.
Thus, combining the two cases implies that $q = p^*$.

Hence, as $q$ was arbitrary, the $\omega$ limit set of $\{p^s\}_{s=1}^\infty$ is
equal to the single point $p^*$.
This implies that $\{p^s\}_{s=1}^\infty$ converges and
that $\lim_{s \to \infty} p^s = p^*$.
Recall that $p^*$ was any arbitrary solution to
\eqref{eq:opt3}, so let $\overline{p}^*$ be any other
solution.
By the above, the sequence $\{p^s\}_{s=1}^\infty$ starting from $p^1 = p_0$
converges to $p^*$ and to $\overline{p}^*$,
so we must have $p^* = \overline{p}^*$.
Therefore, there exists a unique solution $p^*$ to
\eqref{eq:opt3}, and $\lim_{s \to \infty} p^s = p^*$.
\end{proof}

\begin{proof}[Proof of Theorem~\ref{thm:1dim}]
As $J$ is one dimensional and $R$ is open and path connected,
$\partial R$ consists of at most two points, say $\partial R = \{p_1,p_2\}$.
As the only 
subset of $\partial R$ with full measure is $\partial R$ itself, by
Theorem~\ref{thm:Gtot} this implies that $G$ extends to a $C^2$ function
$\tilde{G}$ on an open neighborhood $J'$ of $\partial R$.
By Sard's Theorem \cite[Theorem 6.10]{Lee13},
the set of regular values of $\tilde{G}$ has full measure, so it is
generically true that zero is a regular value of
$\tilde{G}$.
Thus, $d\tilde{G}_p$ is full rank over
$\partial R = \tilde{G}^{-1}(0)$.
Let $r > 0$ such that $\overline{B}_r(p_1), \overline{B}_r(p_2) \subset J'$
and $\overline{B}_r(p_1), \overline{B}_r(p_2)$ are disjoint.
Let $B = \overline{B}_r(p_1) \cup \overline{B}_r(p_2)$
and let $N$ be the interior of $B$, which is an open neighborhood of
$\partial R$.
As $d\tilde{G}_p$ is continuous over $B$ compact, there exist constants
$c_m, c_M$ such that
$c_m \leq |d\tilde{G}_p| \leq c_M$
for all $p \in B$.
As $d\tilde{G}_p$ is full rank over $\partial R$, $|d\tilde{G}_p| > 0$ on
$\partial R$, so shrinking $r$ if necessary implies that $c_m, c_M > 0$.
Let $d > 0$ such that
$k := \f{d}{c_m} < 1$.
As $d\tilde{G}_p$ is continuous over $B$ compact, it is uniformly continuous,
so there
exists $\alpha > 0$ such that $p,q \in B$ with $|p-q| < \alpha$ implies that
$|d\tilde{G}_p-d\tilde{G}_q| < d$.
Shrink $r$ further if necessary so that $p,q$ in one connected component of $B$
implies that $|p-q| < \alpha$.
For any $p \in R \cap B$,
by \eqref{eq:1dim} we have that $\tilde{F}(p) = p - (dG_p)^{-1}G(p)$.
Let $p^*$ be the closest point in $\partial R$ to $p$, which is unique
since $\overline{B}_r(p_1), \overline{B}_r(p_2)$ are disjoint.
By the mean value theorem, there exists $q$
such that $q = tp + (1-t)p^*$ for some $t \in [0,1]$ and
$d\tilde{G}_q(p-p^*) = \tilde{G}(p) - \tilde{G}(p^*) = \tilde{G}(p) = G(p)$,
where the last equality holds since $\tilde{G}(p^*) = 0$ since
$p^* \in \partial R$, and since $\tilde{G}(p) = G(p)$ since $p \in R$.
This implies that
\begin{align*}
  \tilde{F}(p)-p^* &= p-p^* - (dG_p)^{-1}G(p) \\
  &= p-p^* - (dG_p)^{-1}d\tilde{G}_q(p-p^*) \\
  &= p-p^* - (dG_p)^{-1}(d\tilde{G}_q-dG_p + dG_p)(p-p^*) \\  
  &= -(dG_p)^{-1}(d\tilde{G}_q-dG_p)(p-p^*).
\end{align*}
Hence, the above bounds on $|d\tilde{G}_p|$ and
$|d\tilde{G}_p-d\tilde{G}_q|$ yield
$|\tilde{F}(p)-p^*| \leq \f{d}{c_m} |p-p^*| = k |p-p^*|$.
As $p^* \in \partial R = \{p_1,p_2\}$,
$p \in B \cap R$ implies that $|p-p^*| \leq r$
so that $|\tilde{F}(p)-p^*| \leq k r$.
Thus, as $p \in R$ and $R$ is open, $m(p)$ is finite so $|F(p)-p^*| < r$.
Since $\partial R \cap \overline{B}_r(p^*) = \{p^*\}$,
by Lemma~\ref{lem:conv}(a)
the sequence $\{p^s\}_{s=1}^\infty$ converges to $p^*$.
\end{proof}

\begin{proof}[Proof of Theorem~\ref{thm:cont}]
Fix $\hat{p} \in \partial R$.
Let $\eta$ be the unit vector which is orthogonal to $dG_{\hat{p}}$.
For any $\kappa \geq 0$, let $H_\kappa$ denote the hyperplane defined by
$(p-\hat{p})^\intercal \eta = \kappa$, and note that $H_\kappa$ is a
one dimensional $C^\infty$ manifold.
By Lemma~\ref{lem:Rmanifold},
$\partial R$ consists of a finite union of $C^2$ manifolds
$\{M_i\}_{i \in I}$.
As $J$ is two dimensional, each $M_i$ has dimension one or less.
As $C^1$ submanifolds
are generically transverse \cite[Theorem A.3.20]{Ka99},
it is generically true that
$H_\kappa$ is transverse to $M_i$ for each $i \in I$ for generic $\kappa \geq 0$,
including at $\kappa = 0$.
Thus, as $J$ is two dimensional, $H_\kappa \cap M_i$ has dimension zero
(or is empty) for each $i \in I$.
Hence, $H_\kappa \cap \partial R = \cup_{i \in I} H_\kappa \cap M_i$ is a zero
dimensional manifold, so it consists of a countable number of isolated points.
As $H_0 \cap \partial R$ contains $\hat{p}$, there exists $j \in I$ such that
$H_0 \cap M_j$ is nonempty and $H_0$, $M_j$ are transverse.
Thus, by the openness of points of transveral intersection
\cite[Corollary~A.3.18]{Ka99},
since $H_\kappa$ varies $C^\infty$ with $\kappa$, there exists $r > 0$ such that
$\kappa \in [0,r]$ implies that $H_\kappa \cap M_j$ is nonempty, with at least
one point of intersection, call it $p(\kappa)$, varying $C^2$ with $\kappa$ and
approaching $\hat{p}$ as $\kappa \to 0$.
Furthermore, by the proof of Theorem~\ref{thm:Gtot}, $G$ extends to a $C^2$
function $\tilde{G}$ on a neighborhood of every manifold $M_i$ with
codimension one, so assume
that $\hat{p}$ lies in one such manifold, say $M_j$, which is true almost
always (i.e., for generic $\hat{p}$).
By Sard's Theorem \cite[Theorem 6.10]{Lee13}, zero is generically a regular
value of $\tilde{G}$.
As $M_j \subset \tilde{G}^{-1}(0)$, this implies that $d\tilde{G}_p$ is full
rank over $M_j$.
In particular, as $\hat{p} \in M_j$, $d\tilde{G}_{\hat{p}}$ is full rank so it is
nonzero.
As $d\tilde{G}_p$ is continuous and
$d\tilde{G}_{\hat{p}}^\intercal d\tilde{G}_{\hat{p}} = ||d\tilde{G}_{\hat{p}}||_2^2> 0$
, shrinking $r$ if necessary implies that for
$p \in p([0,r])$, $d\tilde{G}_p^\intercal d\tilde{G}_{\hat{p}} > 0$.
As $p(\kappa)$ is continuous, there exists $\hat{r} > 0$ such that
$||\hat{p}+\kappa \eta - p(\kappa)|| \leq \hat{r}$ for all $\kappa \in [0,r]$.
As the closed $\hat{r}$-neighborhood of $p([0,r])$ is compact and
$d\tilde{G}_p$ is continuous, shrinking $r$ further if necessary there
exists $c_m > 0$ such that
\begin{align}
  0 < c_m \leq d\tilde{G}_p^\intercal d\tilde{G}_{\hat{p}} \label{eq:cm2}
\end{align}
for all $p$ in the $\hat{r}$-neighborhood of $p([0,r])$.
This implies that the above bound holds for any $\kappa \in [0,r]$ and any
$p$ such that $||p-p(\kappa)|| \leq ||\hat{p}+\kappa \eta - p(\kappa)||$
as the latter is bounded by $\hat{r}$.
As $p(\kappa)$ is isolated in $M_j \cap H_\kappa$ for all $\kappa \in [0,r]$,
and since the $\hat{r}$-neighborhood of $p([0,r])$ is compact, shrinking $r$
further if necessary implies that this $\hat{r}$-neighborhood does not contain
any points in $M_j$ other than $p([0,r])$ (otherwise, transversality of
$M_j$ and $H_\kappa$ over this compact neighborhood would have to be lost for
new points of intersection to arise arbitarily close to $\hat{p} = p(0)$).
Shrink $r$ further if necessary so that the $\hat{r}$-neighborhood of
$p([0,r])$ does not intersect $M_i$ for any $i \in I$ with $i \neq j$
(this is possible since $p([0,r])$ is compact and $M_i$, $M_j$ are pairwise
disjoint for $i \neq j$).
Choose $d > 0$ such that
\begin{align}
  k := \f{4d||d\tilde{G}_{\hat{p}}||}{c_m} < 1 \label{eq:k2}
\end{align}
As $d\tilde{G}_p$ is continuous over the $\hat{r}$-neighborhood of $p([0,r])$,
which is compact, it is uniformly continuous, so there exists $\alpha > 0$
such that $||p-q|| < \alpha$ implies
\begin{align}
  ||d\tilde{G}_p-d\tilde{G}_q|| < d. \label{eq:d2}
\end{align}
Shrink $r$ further if necessary so that $p, q$ in the $\hat{r}$-neighborhood
of $p([0,r])$ implies that $||p-q|| < \alpha$.

Fix $\kappa \in [0,r]$ and let $p^* = p(\kappa)$.
For any $p$ in the domain of $\tilde{G}$, write
$w_p = d\tilde{G}_p = \begin{bmatrix} a & b \end{bmatrix}^\intercal$ and define
the map $\eta_p = \eta(p) = \begin{bmatrix} -b & a \end{bmatrix}^\intercal$.
Then $\eta(p)$ is orthogonal to $dG_p$ and satisfies $||\eta(p)|| = ||dG_p||$.
Note that $\eta(\hat{p}) = \eta$.
Fix any $p \in \overline{B}_{\hat{r}}(p(\kappa)) \cap H_\kappa$,
where $\overline{B}_{\hat{r}}(p(\kappa))$ is the closed ball of radius $\hat{r}$
centered at $p(\kappa)$.
We have that
 $d\tilde{F}_p =
  \begin{bmatrix}
    w_p^\intercal \\
    \eta^\intercal
  \end{bmatrix}$,
and its inverse is given by
  $d\tilde{F}_p^{-1} =
  \f{1}{w_p^\intercal w_{\hat{p}}}
  \begin{bmatrix}
    w_{\hat{p}} & \eta_p
  \end{bmatrix}$,
which is well-defined by the choice of $p$ due to \eqref{eq:cm2}.
By the multivariate mean value theorem, there exists $q$
such that $q = tp + (1-t)p^*$ for some $t \in [0,1]$ and
$\tilde{F}(p) = \tilde{F}(p) - \tilde{F}(p^*) = d\tilde{F}_q(p-p^*)$
where the first equality follows since $\tilde{F}(p^*) = 0$ because
$p^*$ satisfies $G(p^*) = 0$ and $p^* \in H_\kappa$.
This implies that
\begin{align}
  \begin{split}
  \tilde{F}(p)-p^* &= p-p^* - d\tilde{F}_p^{-1}\tilde{F}(p) \\
  &= p-p^* - \f{1}{w_p^\intercal w_{\hat{p}}}
  \begin{bmatrix}
    w_{\hat{p}} & \eta_p
  \end{bmatrix}
  \begin{bmatrix}
    w_q^\intercal \\
    \eta^\intercal
  \end{bmatrix}
  (p-p^*). \label{eq:Fpstar}
  \end{split}
\end{align}
As $w_{\hat{p}}$ and $\eta = \eta_{\hat{p}}$ are orthogonal, we can decompose
the identity matrix $I$ as
$I = \f{w_{\hat{p}} w_{\hat{p}}^\intercal}{w_{\hat{p}}^\intercal w_{\hat{p}}}
  + \f{\eta \eta^\intercal}{w_{\hat{p}}^\intercal w_{\hat{p}}}$
since, letting $R$ denote rotation of a vector by $90^\circ$, we have that
$\eta = Rw_{\hat{p}}$ so
$\eta^\intercal \eta = w_{\hat{p}}^\intercal R^\intercal R w_{\hat{p}}
= w_{\hat{p}}^\intercal w_{\hat{p}}$ because $R$ is orthogonal so
$R^\intercal = R^{-1}$.
Writing $p-p^* = I(p-p^*)$ and substituting this decomposition into
\eqref{eq:Fpstar} yields
\begin{align*}
  \tilde{F}(p)-p^* &= \left(I -
  \f{w_{\hat{p}}w_q^\intercal + \eta_p \eta^\intercal}{w_p^\intercal w_{\hat{p}}}
  \right) (p-p^*) \\
  &= \f{\left(\f{w_p^\intercal w_{\hat{p}}}{w_{\hat{p}}^\intercal w_{\hat{p}}} -1 + 1
    \right)
 w_{\hat{p}} w_{\hat{p}}^\intercal - w_{\hat{p}} w_q^\intercal}{w_p^\intercal w_{\hat{p}}}
  (p-p^*) \\
  &+ \f{\left(\f{w_p^\intercal w_{\hat{p}}}{w_{\hat{p}}^\intercal w_{\hat{p}}} -1 +1
    \right)
    \eta \eta^\intercal - \eta_p\eta^\intercal}{w_p^\intercal w_{\hat{p}}}
  (p-p^*) \\
  &= \f{\f{(w_p-w_{\hat{p}})^\intercal w_{\hat{p}}}{w_{\hat{p}}^\intercal w_{\hat{p}}}
    w_{\hat{p}}w_{\hat{p}}^\intercal + w_{\hat{p}}(w_{\hat{p}}-w_q)^\intercal}
        {w_p^\intercal w_{\hat{p}}} (p-p^*) \\
  &+ \f{\f{(w_p-w_{\hat{p}})^\intercal w_{\hat{p}}}{w_{\hat{p}}^\intercal w_{\hat{p}}}
    \eta \eta^\intercal + (\eta-\eta_p)\eta^\intercal}
        {w_p^\intercal w_{\hat{p}}} (p-p^*).
\end{align*}
For vectors $u$, $v$, we have $||uv^\intercal ||_2 = ||u||_2||v||_2$.
Thus, by submultiplicity, the Cauchy-Schwartz inequality, since
$||\eta|| = ||w_{\hat{p}}||$, and by \eqref{eq:d2}, \eqref{eq:cm2} we have
\begin{align*}
  ||\tilde{F}(p)-p^*||_2 &\leq
  \f{4d ||w_{\hat{p}}||_2}{c_m}||p-p^*||_2 = k ||p-p^*||_2
\end{align*}
where $k \in (0,1)$ by \eqref{eq:k2}.

Note that, as $(p-\hat{p})^\intercal \eta - \kappa$ is a linear equation (which
defines the hyperplane $H_\kappa$), Newton's method automatically ensures
that $\tilde{F}(p)$ satisfies this equation (i.e., that
$\tilde{F}(p) \in H_\kappa$).
If the initial condition $p^1$ lies in $R \cap H_\kappa \cap B_{\hat{r}}(p^*)$,
then $F(p^s)$ must also lie in $H_\kappa$ for all future iterations as well,
so $H_\kappa$ is forward invariant under $F$.
In general, $\hat{p}+\kappa \eta$ may not lie in $R \cap B_{\hat{r}}(p^*)$,
as will be required to show convergence of the method.
Therefore, to find a suitable starting point, we use a line search
(such as with the bisection or golden section methods) along $H_\kappa$
starting from $\hat{p}+\kappa \eta$.
This procedure yields $p^1 \in R \cap H_\kappa \cap B_{\hat{r}}(p^*)$.
Suppose $p \in R \cap H_\kappa \cap B_{\hat{r}}(p^*)$.
Then $||p-p^*||_2 \leq \hat{r}$ so
$||\tilde{F}(p)-p^*||_2 \leq k ||p-p^*||_2 \leq k\hat{r} < \hat{r}$.
Thus, as $p \in R$ and $R$ is open, $m(p)$ is finite so
$||F(p)-p^*||_2 < \hat{r}$.
Thus, $B_{\hat{r}}(p^*)$ is forward invariant as well.
By the choice of $r$ and $\kappa$, the $\hat{r}$-neighborhood of $p([0,r])$
does not contain any points of $M_j$ other than $p([0,r])$, and also does not
intersect $M_i$ for any $i \neq j$.
Thus, $\partial R \cap H_\kappa \cap B_{\hat{r}}(p^*) = \{p^*\}$.
Therefore, by Lemma~\ref{lem:conv}(b) the sequence $\{p^s\}_{s=1}^\infty$
converges to $p^*$.
\end{proof}

\begin{lemma}
  \label{lem:num}
Assume the conditions of Theorem~\ref{thm:conv}.
Then for any $\epsilon \in (0,G(p_0))$,
the solutions to \eqref{eq:num} are equal to the solutions to
\eqref{eq:opt3}.
\end{lemma}

\begin{proof}[Proof of Lemma~\ref{lem:num}]
As \eqref{eq:num} has the same objective function and a strictly larger
feasible set than \eqref{eq:opt3}, the optimal cost of
\eqref{eq:num} is less than or equal to the optimal cost of
\eqref{eq:opt3}.
Assume towards a contradiction that there exists a solution $p^*$ to
\eqref{eq:num} that is not a solution to \eqref{eq:opt3}.
If $G(p^*) = \epsilon$ then $p^*$ would be feasible
for \eqref{eq:opt3}, but suboptimal since by assumption it is
not a solution to \eqref{eq:opt3}.
However, this contradicts that the optimal cost of
\eqref{eq:opt3} is no less than that of \eqref{eq:num}.
So, we must have $G(p^*) < \epsilon$.
Let $\gamma:[0,1] \to J$ be a length-minimizing geodesic with respect to the
metric $d_P$ such that $\gamma(0) = p_0$, $\gamma(1) = p^*$, and $\gamma$
is continuous.
As $G$ extends to a continuous function on a neighborhood of $\partial R$
by Theorem~\ref{thm:Gtot}, $G \circ \gamma$ is continuous and satisfies
$(G \circ \gamma)(0) = G(p_0) > \epsilon$ and
$(G \circ \gamma)(1) = G(p^*) < \epsilon$.
Thus, by the intermediate value theorem there must exist $s \in (0,1)$
such that $(G \circ \gamma)(s) = \epsilon$.
Define $\hat{p} = \gamma(s)$ and note that $G(\hat{p}) = \epsilon$
so $\hat{p}$ is feasible for \eqref{eq:num}.
As $\gamma$ is length-minimizing and $s < 1$,
$d_P(p_0,\hat{p}) = d_P(\gamma(0),\gamma(s)) < d_P(\gamma(0),\gamma(1))
  = d_P(p_0,p^*)$
which, since $\hat{p}$ is feasible for \eqref{eq:num}, contradicts that
$p^*$ is a solution to \eqref{eq:num}.
Thus, every solution to \eqref{eq:num} is a solution to
\eqref{eq:opt3}.
In particular, as \eqref{eq:num} has the same objective function and a strictly
larger feasible set than \eqref{eq:opt3}, this implies that
they have the same solutions.
\end{proof}

\begin{proof}[Proof of Theorem~\ref{thm:opt}]
By Lemma~\ref{lem:Rmanifold}, $\partial R = \cup_{i \in I} M_i$ where each $M_i$ is a $C^3$
embedded submanifold, $I$ is finite, and $M_i$ and $M_j$ are disjoint for
all $j \neq i$ in $I$.
Therefore, for each $i \in I$ there exists a tubular neighborhood $N_i$ of
$M_i$ in $J$ and a $C^3$ retraction $r_i:V_i \to M_i$ such that for
every $p_0 \in V_i$,
$r_i(p_0)$ is the unique closest point in $M_i$ to $p_0$
\cite[Problem~6-5]{Lee13} where distance is measured by the metric $d_P$.
Note that, for any $p_0 \in J$, $d_P(p_0,\partial R) = d_P(p_0,M_i)$
for some $i \in I$ since $\partial R = \cup_{i \in I} M_i$, and $r_i(p_0)$
is the unique closest point in $M_i$ to $p_0$ (in particular,
$d_P(p_0,M_i) = d_P(p_0,r_i(p_0))$).

Let $N = \cup_{i \in I} V_i$, and note that $N$ is an open neighborhood of
$\partial R$.
Let $\hat{V}$ consist of the set of $p_0$ such that
$d_P(p_0,\partial R) = d_P(p_0,M_i) = d_P(p_0,M_j)$ for at least two distinct
$i, j$, and define $V = N - \hat{V}$.
We claim that $V$ is open and dense (in particular, generic) in $N$.
For, let $p_0 \in V$.
Then there exists a unique $i \in I$
such that $d_P(p_0,\partial R) = d_P(p_0,M_i) = d_P(p_0,r_i(p_0))$,
so $r_i(p_0)$ is the unique closest point to $p_0$ in $\partial R$.
By uniqueness, this implies that
$d_P(p_0,\partial R) = d_P(p_0,r_i(p_0)) = d_P(p_0,M_i)
< d_P(p_0,M_j) = d_P(p_0,r_j(p_0))$
for all $j \neq i$ in $I$.
As each $r_i$ is continuous and $d_P$ is continuous, there exists an open
neighborhood $J'$ of $p_0$ in $J$ such that $p \in J'$ implies that
$d_P(p,r_i(p)) < d_P(p,r_j(p))$ for all $j \neq i$ in $I$.
In turn, this implies that
$d_P(p,M_i) = d_P(p,r_i(p)) < d_P(p,r_j(p)) = d_P(p,M_j)$
for all $j \neq i$ in $I$.
Thus, $d_P(p,\partial R) = d_P(p,M_i) = d_P(p,r_i(p))$, and $r_i(p)$ is
the unique closest point in $\partial R$ to $p$.
Therefore, $p \in V$ for all $p \in J'$, so $V$ is open.

Next, suppose $p_0 \in N - V$ and let $\epsilon > 0$.
Let $I' \subset I$ be such that for $i \in I$,
$d_P(p_0,\partial R) = d_P(p_0,M_i)$ if and only if $i \in I'$.
By the choice of $p_0$, $I'$ must contain at least two elements.
Fix any $i \in I'$.
As $d_P(p_0,M_i) = d_P(p_0,r_i(p_0))$, and by construction of the tubular
neighborhood $V_i$, there exists a unique length-minimizing geodesic
$\gamma:[0,d_P(p_0,M_i)] \to J$ such that
$\gamma(0) = p_0$, $\gamma(1) = r_i(p_0)$,
$d_P(\gamma(t),\gamma(s)) = s-t$ for every $t < s \in [0,d_P(p_0,M_i)]$,
and $r_i$ is constant over $\gamma$ (i.e. $r_i(\gamma(s)) = r_i(p_0)$ for
all $s \in [0,d_P(p_0,M_i)]$).
Consider $\gamma(\epsilon)$.
First of all,
$d_P(p_0,\gamma(\epsilon)) = d_P(\gamma(0),\gamma(\epsilon)) = \epsilon$.
Furthermore,
$d_P(\gamma(\epsilon),M_i) = d_P(\gamma(\epsilon),r_i(\gamma(\epsilon)))
= d_P(\gamma(\epsilon),\gamma(d_P(p_0,M_i))) = d_P(p_0,M_i) - \epsilon$
so $\gamma(\epsilon)$ is $\epsilon$ closer to $M_i$ than $p_0$.
Consider any $j \in I'$ with $j \neq i$.
If $\gamma(\epsilon)$ lies on the same length-minimizing geodesic as the one
from $r_j(p_0)$ to $p_0$ (i.e. the geodesic from $M_j$ to $p_0$),
then by uniqueness this must be the same geodesic (extended to larger $s$) as
$\gamma$.
In this case, we must have
$d_P(\gamma(\epsilon),M_j) = d_P(p_0,M_j) \pm \epsilon = d_P(p_0,M_i) \pm
\epsilon$.
But, if
$d_P(\gamma(\epsilon),M_j)) = d_P(p_0,M_i) - \epsilon$ then
$d_P(\gamma(\epsilon),r_j(p_0)) = d_P(\gamma(\epsilon),M_j))
= d_P(p_0,M_i) - \epsilon = d_P(\gamma(\epsilon),r_i(p_0))$
and
$d_P(p_0,r_j(p_0)) = d_P(p_0,r_i(p_0))$
so $r_i(p_0), r_j(p_0)$ both lie on the same geodesic $\gamma$ and are
equidistant from $p_0$ and $\gamma(s)$, which are distinct and also lie on
$\gamma$.
This is only possible if $r_i(p_0) = r_j(p_0)$, which contradicts that
$M_i$ and $M_j$ are disjoint.
Thus, in case $\gamma(\epsilon)$ lies on the same geodesic as the one from
$p_0$ to $r_j(p_0)$, we must have
$d_P(\gamma(\epsilon),M_j) = d_P(p_0,M_i) + \epsilon >
d_P(\gamma(\epsilon),M_i)$.

In case $\gamma(\epsilon)$ does not lie on the same geodesic as the one
from $p_0$ to $r_j(p_0)$ (which is the unique length-minimizing geodesic
from $p_0$ to $M_j$),
$d_P(\gamma(\epsilon),M_j) > d_P(p_0,M_j) - \epsilon$
by uniqueness of the length-minimizing geodesic.
This implies that
$d_P(\gamma(\epsilon),M_j) > d_P(p_0,M_j) - \epsilon
  = d_P(p_0,M_i) - \epsilon
  = d_P(\gamma(\epsilon),M_i)$.
Finally, for any $j \in I - I'$,
$d_P(p_0,r_i(p_0)) = d_P(p_0,M_i) < d_P(p_0,M_j) = d_P(p_0,r_j(p_0))$.
So by continuity of $r_j$ and $d_P$, as $d_P(\gamma(\epsilon),p_0) = \epsilon$,
for $\epsilon$ sufficiently small it will hold that
$d_P(\gamma(\epsilon),M_i) = d_P(\gamma(\epsilon),r_i(\gamma(\epsilon)))
< d_P(\gamma(\epsilon),r_j(\gamma(\epsilon))) = d_P(\gamma(\epsilon),M_j)$.
Combining the above cases implies that
$d_P(\gamma(\epsilon),M_i) < d_P(\gamma(\epsilon),M_j)$
for all $j \neq i$ in $I$.
The implies that $d_P(\gamma(\epsilon),\partial R) = d_P(\gamma(\epsilon),M_i)$
and $d_P(\gamma(\epsilon),\partial R) < d_P(\gamma(\epsilon),M_j)$
for all $j \neq i$ in $I$, so $\gamma(\epsilon) \in V$.
As $\gamma(\epsilon) \to p_0$ as $\epsilon \to 0$, this implies that
$V$ is dense in $N$.

Let $p_0 \in V$.
Then there exists a unique $i \in I$ such that $r_i(p_0)$ is the unique
closest point in $\partial R$ to $p_0$.
As any solution to \eqref{eq:opt2} is a closest
point on $\partial R$ to $p_0$ in terms of the metric $d_P$, this implies that
$r_i(p_0)$ is the unique solution to \eqref{eq:opt2} for $p_0$.
Furthermore, by the proof of openness of $V$ above, there exists
an open neighborhood $J'$ of $p_0$ such that $J' \subset V$
and for $p \in J'$, $d_P(p,\partial R) = d_P(p,r_i(p))$,
and thus $r_i(p)$ is the unique solution to \eqref{eq:opt2}.
As each $r_i$ is $C^3$, this implies that the solution varies $C^3$ over $J'$.
Hence, as $p_0 \in V$ was arbitary, this implies that the solution
varies $C^3$ over $V$.
\end{proof}


\begin{proof}[Proof of Theorem~\ref{thm:conv}]
By Lemma~\ref{lem:Rmanifold}, $\partial R = \cup_{i \in I} M_i$ a finite union
of $C^3$ submanifolds.
By Lemma~\ref{lem:onedim} and its proof, the set of $p^* \in \partial R$
such that $W^s(X^*(p^*))$ has codimension one is a subset of the
$\{M_i\}_{i \in I}$, call them $Y := \{M_i\}_{i \in \hat{I}}$ for some
$\hat{I} \subset I$, with full Lebesgue measure in $\partial R$.
By Theorem~\ref{thm:Gtot}, for each $p^* \in Y$ there exists an open
neighborhood $N_{p^*}$ of 
$p^*$ such that $G$ extends to a $C^2$
function $\tilde{G}$ on $N_{p^*}$.
For each $i \in \hat{I}$, let $\hat{N}_i = \cup_{p^* \in M_i} N_{p^*}$, and note
$\hat{N}_i$ is an open neighborhood of $M_i$ 
on which $\tilde{G}$
exists and is $C^2$.
By the proof of Theorem~\ref{thm:opt}, for each $i \in I$ there exists an
open neighborhood $V_i$ of $M_i$ such that for generic $p_0 \in V_i$
there exists a unique solution to \eqref{eq:opt2}.
For each $i \in \hat{I}$, let $N_i = \hat{N}_i \cap V_i$, and for each
$i \in I - \hat{I}$ let $N_i = V_i$.
Let $N = \cup_{i \in I} N_i$.
Then $N$ is an open neighborhood of $\partial R$.
Let $p_0 \in N$.
By construction, there exists a unique solution $\hat{p}$ to \eqref{eq:opt2}
for $p_0$, which implies that $\hat{p} \in \partial R$.
Assume that $\hat{p} \in Y$, which is generic since $Y$ has full measure
in $\partial R$.
This implies that $\hat{p} \in \hat{N}_i =: J'$ for some $i \in \hat{I}$.
Then $\tilde{G}$ is well-defined and $C^2$ over $J'$ by construction.

For any $r > 0$, define
$\overline{B}_r(\hat{p}) = \{p \in J:~d_P(p,\hat{p}) \leq r\}$
to be the closed ball of radius $r$ centered at $\hat{p}$.
Then there exists $r > 0$ sufficiently small such that
$\overline{B}_r(\hat{p}) \subset J'$.
As $\tilde{G}$ is $C^2$ over $\overline{B}_r(\hat{p})$,
$d\tilde{G}_p$ is $C^1$ over $\overline{B}_r(\hat{p})$ compact.
Thus, $d\tilde{G}_p$ is $L$-Lipschitz 
over
$\overline{B}_r(\hat{p})$, 
and 
satisfies
\begin{align}
  ||d\tilde{G}_p-d\tilde{G}_q||_P \leq L||p-q||_P  \label{eq:lip}
\end{align}
for all $p, q \in \overline{B}_r(\hat{p})$.
As $\tilde{G}$ is $C^2$, by Sard's Theorem \cite[Theorem~6.10]{Lee13} zero is
generically
a regular value of $\tilde{G}$, so $d\tilde{G}_p$ is full rank over
$\tilde{G}^{-1}(0)$.  
As $\hat{p} \in \tilde{G}^{-1}(0)$,
$d\tilde{G}_{\hat{p}}$ is full rank, so $||d\tilde{G}_{\hat{p}}||_P > 0$.
As $||d\tilde{G}_p||_P$ is continuous over $\overline{B}_r(\hat{p})$ compact,
it achieves a maximum $c_M$ and a minimum $c_m$ over $\overline{B}_r(\hat{p})$.
As $||d\tilde{G}_{\hat{p}}||_P > 0$, shrinking $r$ yields 
$c_m > 0$.
Thus, 
\begin{align}
    0 < c_m \leq ||d\tilde{G}_p||_P \leq c_M \label{eq:cm}
\end{align}
for all $p \in \overline{B}_r(\hat{p})$.
Furthermore, as $d\tilde{G}_{\hat{p}}$ is full rank and full rank is an open
condition, shrinking $r$ further if necessary we may have that $d\tilde{G}_p$
is full rank for all $p \in \overline{B}_r(\hat{p})$.
As $P$ is symmetric positive definite, its eigenvalues are real and positive.
Let $\lambda_{max}$ and $\lambda_{min}$ denote the maximum and minimum
eigenvalues of $P$, respectively.
Choose any $\epsilon \in (0,G(p_0))$.
Shrink $r$ if necessary and choose $d > 0$ such that
\begin{align}
  k :=  \f{\lambda_{max}^2c_M}{\lambda_{min}^2c_m^2}(d + 2rL) &< 1 \label{eq:k} \\
  \left(1 - \f{\lambda_{max}^2c_M}{\lambda_{min}^2c_m^2} d\right)\epsilon
  - \f{r}{\lambda_{min}}\left(1 + \f{\lambda_{max}^2c_M^2}{\lambda_{min}^2c_m^2}
  \right)d &> 0. \label{eq:Gpos}
\end{align}
As $d\tilde{G}_p$ is continuous over $\overline{B}_r(\hat{p})$ compact,
it is uniformly continuous,
so there exists $\alpha > 0$ such that for $||p-q||_P < \alpha$,
\begin{align}
  ||d\tilde{G}_p-d\tilde{G}_q||_P < d.
  \label{eq:d}
\end{align}
Shrink $r$ if necessary so that $p, q \in \overline{B}_r(\hat{p})$ implies that
$||p-q||_P < \alpha$.
As for 
Theorem~\ref{thm:conv}, assume that
$d_P(p_0,\partial R) < \f{r}{3}$.

Consider now the 
program 
in \eqref{eq:opt3}.
As $G^{-1}(\epsilon)$ is closed and $\{p_0\}$ is compact, there exists at least
one point $p^* \in G^{-1}(\epsilon)$ such that
$d_P(p_0,G^{-1}(\epsilon)) = d_P(p_0,p^*)$.
Thus, $p^*$ is a solution to \eqref{eq:opt3}, so it is also a
local optimum. Therefore, since $d\tilde{G}_{p^*}$ is full rank, by the Lagrange
multiplier theorem there exists a unique $\lambda^* \neq 0$ such that
$df_{p^*} = -\lambda^* dG_{p^*}$, where $f(p) = \f{1}{2}(p-p_0)^\intercal P(p-p_0)$.
As $df_p = P(p-p_0)$, this implies that
\begin{align}
  P(p^*-p_0) &= -\lambda^* dG_{p^*}, \quad G(p^*) = \epsilon. \label{eq:L}
\end{align}
Let $\gamma(t) = (1-t)p_0 + t\hat{p}$.
As $\epsilon \in (0,G(p_0)) = (G(\gamma(1)),G(\gamma(0)))$,
by the intermediate value theorem there exists $t \in (0,1)$ such that
$G(\gamma(t)) = \epsilon$.
This implies that
$d_P(p_0,p^*) = d_P(p_0,G^{-1}(\epsilon)) \leq d_P(p_0,\gamma(t))
< d_P(p_0,\hat{p}) = d_P(p_0,\partial R) < \f{r}{3}$.
Hence, for any $x \in B_{d_P(p_0,p^*)}(p^*)$,
$d_P(x,\hat{p}) \leq d_P(x,p^*) + d_P(p^*,p_0) + d_P(p_0,\hat{p})
= \f{r}{3} + \f{r}{3} + \f{r}{3} = r$.
So, $B_{d_P(p_0,p^*)}(p^*) \subset B_r(\hat{p})$.

Consider now the quadratic program given by \eqref{eq:opt4}
whose solution is $\tilde{F}(p^s)$ for $p^s$ the current iteration.
The Lagrangian for this constrained optimization is given by
\begin{align*}
  \mathcal{L}(p^{s+1},\lambda) &= \f{1}{2} (p^{s+1}-p_0)^\intercal P(p^{s+1}-p_0)
  \\&+ \lambda(G(p^s) + dG_{p^s}^\intercal(p^{s+1}-p^s) - \epsilon).
\end{align*}
Stationary points of the Lagrangian are given by
\begin{align*}
  0 &= \nabla \mathcal{L}(p^{s+1},\lambda)
  = \begin{bmatrix}
    P(p^{s+1}-p_0) + \lambda dG_{p^s} \\
    G(p^s) + dG_{p^s}^\intercal(p^{s+1}-p^s) - \epsilon
  \end{bmatrix}.
\end{align*}
Write
$p = p^s$, 
$\tilde{F}(p) = p^{s+1}$,
and $w_p = d\tilde{G}_{p^s}$.
Then 
a stationary point of the Lagrangian
is given by
\begin{align*}
  \begin{bmatrix}
    P & w_p \\
    w_p^\intercal & 0    
  \end{bmatrix}
  \begin{bmatrix}
    \tilde{F}(p) \\
    \lambda
  \end{bmatrix}
  =
  \begin{bmatrix}
    Pp_0 \\
    \epsilon -G(p) + w_p^\intercal p
  \end{bmatrix}.
\end{align*}
For any $p \in \overline{B}_r(\hat{p})$, $w_p = d\tilde{G}_p \neq 0$,
so the above equation
has a unique solution $(\tilde{F}(p),\lambda)$ given by
\begin{align*}
  \begin{bmatrix}
    \tilde{F}(p) \\
    \lambda
  \end{bmatrix}
  &= \begin{bmatrix}
     P^{-1}(w_p^\intercal P^{-1}w_pI - w_pw_p^\intercal P^{-1}) & P^{-1}w_p \\
    w_p^\intercal P^{-1} & -1
  \end{bmatrix} \\
  &*\f{1}{w_p^\intercal P^{-1}w_p}
  \begin{bmatrix}
    Pp_0 \\
    \epsilon -G(p) + w_p^\intercal p
  \end{bmatrix}.
\end{align*}
for each $p \in \overline{B}_r(\hat{p})$.
This implies that
\begin{align}
  \tilde{F}(p) =
  p_0 + \f{P^{-1}w_pw_p^\intercal}{w_p^\intercal P^{-1}w_p}(p-p_0)
  - \f{P^{-1}w_p}{w_p^\intercal P^{-1}w_p}(G(p)-\epsilon). \label{eq:tf}
\end{align}

Fix $p \in \overline{B}_r(\hat{p})$.
By the multivariate mean value theorem, there exists $q$
with
$q = tp + (1-t)p^*$ for some $t \in [0,1]$ and
$(G(p)-\epsilon) - (G(p^*)-\epsilon) = w_q^\intercal(p-p^*)$.
As $G(p^*) = \epsilon$, this implies that
$G(p) - \epsilon = w_q^\intercal(p-p^*)$.
Substituting 
this into \eqref{eq:tf},
writing $p_0-p^* = p_0-p+p-p^*$, and writing $p-p_0 = p-p^* + p^*-p_0$,
we compute
\begin{align}
  \tilde{F}(p)-p^*
  &= \left(I + \f{P^{-1}w_p(w_p-w_q)^\intercal}{w_p^\intercal P^{-1}w_p}\right)(p-p^*) \nonumber
  \\&+ p_0-p + \f{P^{-1}w_pw_p^\intercal}{w_p^\intercal P^{-1}w_p}(p^*-p_0).
  \label{eq:fpps}
\end{align}
The last term in this expression can be written,
using \eqref{eq:L}, as
\begin{align*}
  & \f{P^{-1}w_pw_p^\intercal}{w_p^\intercal P^{-1}w_p}(p^*-p_0) 
  \\& 
  = -\lambda^*\f{w_p^\intercal P^{-1}w_{p^*}}{w_p^\intercal P^{-1}w_p}P^{-1}(w_{p^*}-w_{p^*}+w_p) \\
  & 
  = \f{w_p^\intercal P^{-1}(w_p-w_p+w_{p^*})}{w_p^\intercal P^{-1}w_p}(p^*-p_0)
  \\& 
  + \f{w_p^\intercal(p^*-p_0)}{w_p^\intercal P^{-1}w_p}P^{-1}(w_p-w_{p^*}) \\
  & 
  = p^*-p_0
  + \f{w_p^\intercal P^{-1}(w_{p^*}-w_p)}{w_p^\intercal P^{-1}w_p}(p^*-p_0) \\
  & 
  + \f{w_p^\intercal(p^*-p_0)}{w_p^\intercal P^{-1}w_p}P^{-1}(w_p-w_{p^*}).
\end{align*}
Substituting this expression back into \eqref{eq:fpps},
taking the norm $||\cdot ||_P$, and applying \eqref{eq:lip}, \eqref{eq:cm},
and \eqref{eq:d} we obtain
\begin{align*}
  ||\tilde{F}(p)-p^*||_P 
  &\leq
  \f{|(w_p-w_q)^\intercal(p-p^*)|}{|w_p^\intercal P^{-1}w_p|}  ||P^{-1}w_p||_P \\
  & \mkern-100mu + \f{|w_p^\intercal P^{-1}(w_{p^*}-w_p)|}{|w_p^\intercal P^{-1}w_p|}||p^*-p_0||_P \\
  & \mkern-100mu
  + \f{|w_p^\intercal(p^*-p_0)|}{|w_p^\intercal P^{-1}w_p|}||P^{-1}(w_p-w_{p^*})||_P \\
  & \mkern-100mu 
  \leq \f{\f{1}{\lambda_{min}^2}dc_M}{\f{1}{\lambda_{max}^2}c_m^2}||p-p^*||_P
  + \f{\f{1}{\lambda_{min}^2}c_MLr}{\f{1}{\lambda_{max}^2}c_m^2}||p-p^*||_P \\
  & \mkern-100mu
  + \f{\f{1}{\lambda_{min}^2}c_M rL}{\f{1}{\lambda_{max}^2}c_m^2}||p-p^*||_P \\
  & \mkern-100mu =
  \f{\lambda_{max}^2c_M}{\lambda_{min}^2c_m^2}(d + 2rL)||p-p^*||_P = k||p-p^*||_P
\end{align*}
where we have used \eqref{eq:k} and that
\begin{align}
  &|v^\intercal w| 
  = |\langle v,P^{-1}w \rangle_P| 
  \leq ||v||_P||P^{-1}w||_P 
  \leq \f{1}{\lambda_{min}} ||v||_P||w||_P \label{eq:abs1} \\
  &|c^\intercal P^{-1}c| = |\langle P^{-1}c,P^{-1}c\rangle_P|
  = ||P^{-1}c||_P^2
  \geq \f{1}{\lambda_{max}^2} ||c||_P^2 \label{eq:abs2}
\end{align}
where the former follows by the Cauchy-Schwartz inequality.
Thus, for any $p \in \overline{B}_r(\hat{p})$, we have
\begin{align}
  ||\tilde{F}(p)-p^*||_P \leq k||p-p^*|| \label{eq:k3}
\end{align}
where $k \in (0,1)$ by 
\eqref{eq:k}.
As $B_{d_P(p_0,p^*)}(p^*) \subset B_r(\hat{p})$ and $p^*$ was arbitrary,
the same inequality holds for any solution $p^*$ to
\eqref{eq:opt3} and any $p$ with
$d_P(p,p^*) \leq d_P(p_0,p^*)$.

Next, for any solution $p^*$ to \eqref{eq:opt3},
let $p \in \partial R$ with $d_P(p,p^*) \leq d_P(p_0,p^*)$.
Then $p \in \partial R$ implies that $\tilde{G}(p) = 0$.
By Theorem~\ref{thm:Gtot}, $R \cap J' = \{p:\tilde{G}(p) > 0\} \cap J'$.
By \eqref{eq:k3},
\begin{align*}
d_P(\hat{p},\tilde{F}(p)) &\leq
d_P(\hat{p},p_0) + d_P(p_0,p^*) + d_P(p^*,\tilde{F}(p)) \\
&< d_P(\hat{p},p_0) + d_P(p_0,p^*) + d_P(p,p^*) \\
&\leq d_P(\hat{p},p_0) + 2d_P(p_0,p^*)
\leq \f{r}{3} + \f{2r}{3} = r,
\end{align*}
so $\tilde{F}(p) \in B_r(\hat{p}) \subset J'$.
So, to show that $\tilde{F}(p) \in R$ it suffices to show that
$\tilde{G}(\tilde{F}(p)) > 0$.
By the multivariate mean value theorem, there exists $q$
such that $q = t\tilde{F}(p) + (1-t)p$ for some $t \in [0,1]$ and
$\tilde{G}(\tilde{F}(p)) - \tilde{G}(p) = w_q^\intercal(\tilde{F}(p)-p)$.
As $\tilde{G}(p) = G(p) = 0$, this implies that
$\tilde{G}(\tilde{F}(p)) = w_q^\intercal(\tilde{F}(p)-p)$.
Substituting Eq.~\ref{eq:tf} for $\tilde{F}(p)$, noting that $G(p) = 0$,
and writing $w_q = w_p - w_p + w_q$, we obtain
\begin{align*}
  \tilde{G}(\tilde{F}(p)) &= 
  \f{(w_p-w_p+w_q)^\intercal P^{-1}w_p}{w_p^\intercal P^{-1}w_p}(w_p^\intercal(p-p_0)+\epsilon) \\
  & -w_q^\intercal(p-p_0)\\
  & \mkern-50mu
  = \left(1 + \f{(w_q-w_p)^\intercal P^{-1}w_p}{w_p^\intercal P^{-1}w_p}\right)\epsilon
  \\& \mkern-50mu - \f{(w_p-w_q)^\intercal P^{-1}w_p}{w_p^\intercal P^{-1}w_p}w_p^\intercal(p-p_0) 
  - (w_q-w_p)^\intercal(p-p_0) \\
  & \mkern-50mu \geq
  \left(1 - \f{\lambda_{max}^2c_M}{\lambda_{min}^2c_m^2}d \right)\epsilon
  - \f{r}{\lambda_{min}}d - \f{\lambda_{max}^2c_M^2}{\lambda_{min}^2c_m^2}
  \f{r}{\lambda_{min}}d \\
  & \mkern-50mu
  \geq \left(1 - \f{\lambda_{max}^2c_M}{\lambda_{min}^2c_m^2}d \right)\epsilon
  - \f{r}{\lambda_{min}}\left(1 + \f{\lambda_{max}^2c_M^2}{\lambda_{min}^2c_m^2}
  \right)d > 0
\end{align*}
using \eqref{eq:abs1}-\eqref{eq:abs2}, \eqref{eq:Gpos}, and 
$d_P(p,p_0) \leq \f{2r}{3} < r$.
Thus, $p \in \partial R$ with $d_P(p,p^*) \leq d_P(p_0,p^*)$ for any
solution $p^*$ to \eqref{eq:opt3} implies 
$\tilde{G}(\tilde{F}(p)) > 0$ and, thus, that $\tilde{F}(p) \in R$.
Combining the 
above, by Lemma~\ref{lem:conv}(c) the sequence
$\{p^s\}_{s=1}^\infty$ 
converges to $p^*$.
\end{proof}


\section{Conclusion}\label{sec:conc}


This work provided theoretical guarantees for recent algorithms which
efficiently and non-conservatively compute safety margins for vulnerability
assessment.
A function $G$ was defined to be the reciprocal of the supremum over
time of the norms of the trajectory sensitivities.
Then, a characterization of the recovery boundary in terms of trajectory
sensitivities was provided in terms of $G$.
In particular, it was then shown that $G$ is strictly positive over the recovery
region, and for a generic parameter value $p^*$ in the
recovery boundary, $G(p^*) = 0$, $G$ is continuous at $p^*$, and $G$ extends to
a $C^2$ function over a neighborhood of $p^*$.

Next, this characterization was used to show the following results under generic
assumptions and for initial parameter values sufficiently
close to the recovery boundary.
Well-posedness and convergence
guarantees were provided for algorithms which use $G$ to find a point on the
recovery boundary
in one dimensional parameter space, and to numerically trace the recovery
boundary in two dimensional parameter space.
It was then shown that the closest point on the recovery boundary
to a nominal parameter value exists, is unique, and depends smoothly on the
nominal value.
In turn, this was then used to provide well-posedness and convergence guarantees
for an algorithm which uses $G$ to find the closest point on the recovery
boundary in arbitrary dimensional parameter space in order to compute the
safety margin.
These theoretical guarantees ensure that these recent algorithms can be
successfully applied to a large class of nonlinear systems.

\bibliographystyle{ieeetr}
\bibliography{refs}

\begin{thebibliography}{10}

\bibitem{Ch88}
H.-D. Chiang, M.~W. Hirsch, and F.~F. Wu, ``Stability regions of nonlinear
  autonomous dynamical systems,'' {\em IEEE Transactions on Automatic Control},
  vol.~33, no.~1, pp.~16--27, 1988.

\bibitem{Ta08}
W.~Tan and A.~Packard, ``Stability region analysis using polynomial and
  composite polynomial lyapunov functions and sum-of-squares programming,''
  {\em IEEE Transactions on Automatic Control}, vol.~53, no.~2, pp.~565--571,
  2008.

\bibitem{Kh02}
H.~K. Khalil, {\em Nonlinear Systems}.
\newblock Pearson Education, Inc., 3~ed., 2002.

\bibitem{El17}
A.~El-Guindy, D.~Han, and M.~Althoff, ``Estimating the region of attraction via
  forward reachable sets,'' {\em American Control Conference (ACC)},
  pp.~1263--1270, 2017.

\bibitem{Zu64}
V.~I. Zubov, {\em Methods of A. M. Lyapunov and their Application}.
\newblock 1964.

\bibitem{Va85}
A.~Vannelli and M.~Vidyasagar, ``Maximal lyapunov functions and domains of
  attraction for autonomous nonlinear systems,'' {\em Automatica}, vol.~21,
  no.~1, pp.~69--80, 1985.

\bibitem{Ca01}
F.~Camilli, L.~Grune, and F.~Wirth, ``A generalization of zubov's method to
  perturbed systems,'' {\em SIAM Journal on Control and Optimization}, vol.~40,
  no.~2, pp.~496--515, 2001.

\bibitem{Pa98}
A.~Paice and F.~Wirth, ``Robustness analysis of domains of attraction of
  nonlinear systems,'' in {\em Proceedings of the Mathematical Theory of
  Networks and Systems}, pp.~353--356, 1998.

\bibitem{To10}
U.~Topcu, A.~K. Packard, P.~Seiler, and G.~J. Balas, ``Robust
  region-of-attraction estimation,'' {\em IEEE Transactions on Automatic
  Control}, vol.~55, no.~1, pp.~137--142, 2010.

\bibitem{Ch13}
G.~Chesi, ``Rational lyapunov functions for estimating and controlling the
  robust domain of attraction,'' {\em Automatica}, vol.~49, no.~4,
  pp.~1051--1057, 2013.

\bibitem{Fi24}
M.~W. Fisher and I.~A. Hiskens, ``Determining disturbance recovery conditions
  by inverse sensitivity minimization,'' 2025.
\newblock Under review. Preprint available on arXiv.

\bibitem{Hi00}
I.~A. Hiskens and M.~A. Pai, ``Trajectory sensitivity analysis of hybrid
  systems,'' {\em IEEE Transactions on Circuits and Systems - I: Fundamental
  Theory and Applications}, vol.~47, no.~2, pp.~204--220, 2000.

\bibitem{Hi76}
M.~W. Hirsch, {\em Differential Topology}, vol.~33 of {\em Graduate Texts in
  Mathematics}.
\newblock Springer-Verlag, 1976.

\bibitem{Lee13}
J.~M. Lee, {\em Introduction to Smooth Manifolds}.
\newblock Graduate Texts in Mathematics, Springer, 2~ed., 2013.

\bibitem{Fi23}
M.~W. Fisher and I.~A. Hiskens, ``Stability of the nonwandering set in the
  region of attraction boundary under perturbations with application to
  vulnerability assessment,'' {\em SIAM Journal on Applied Dynamical Systems},
  vol.~22, no.~4, pp.~3390--3430, 2023.

\bibitem{Fi22}
M.~W. Fisher and I.~A. Hiskens, ``Hausdorff continuity of region of attraction
  boundary under parameter variation with application to disturbance
  recovery,'' {\em SIAM Journal on Applied Dynamical Systems}, vol.~21, no.~1,
  pp.~327--365, 2022.

\bibitem{Se85}
G.~R. Sell, ``Smooth linearization near a fixed point,'' {\em American Journal
  of Mathematics}, vol.~107, no.~5, pp.~1035--1091, 1985.

\bibitem{Fi18c}
M.~W. Fisher and I.~A. Hiskens, ``Numerical computation of critical system
  recovery parameter values by trajectory sensitivity maximization,'' {\em 58th
  Conference on Decision and Control (CDC)}, pp.~8000--8006, 2019.

\bibitem{Ka99}
A.~Katok and B.~Hasselblatt, {\em Introduction to the Modern Theory of
  Dynamical Systems}, vol.~54 of {\em Encyclopedia of Mathematics and its
  Applications}.
\newblock Cambridge University Press, 1999.

\bibitem{Fi19b}
M.~W. Fisher and I.~A. Hiskens, ``Parametric dependence of large disturbance
  response for vector fields with event-selected discontinuities,'' {\em 18th
  European Control Conference (ECC)}, pp.~166--173, 2019.

\bibitem{Tr96}
R.~T. Treinen, V.~Vittal, and W.~Kliemann, ``An improved technique to determine
  the controlling unstable equilibrium point in a power system,'' {\em IEEE
  Transactions on Circuits and Systems - I: Fundamental Theory and
  Applications}, vol.~43, no.~4, pp.~313--323, 1996.

\bibitem{Hu48}
W.~Hurewicz and H.~Wallman, {\em Dimension Theory}.
\newblock Princeton University Press, 1948.

\bibitem{Pa82}
J.~Jacob~Palis and W.~de~Melo, {\em Geometric Theory of Dynamical Systems}.
\newblock Springer-Verlag, 1982.

\end{thebibliography}

\begin{IEEEbiography}
  [{\includegraphics[width=1in,height=1.25in,clip,keepaspectratio]
      {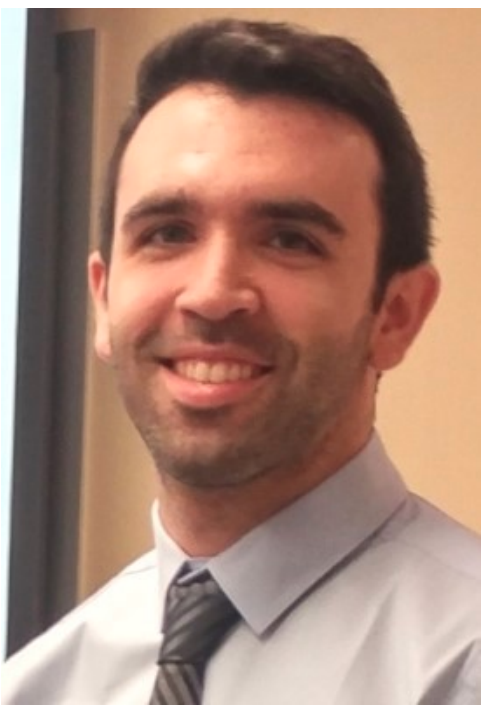}}]{Michael W. Fisher} is an Assistant Professor
  in the Department of Electrical and Computer Engineering at the University of
  Waterloo, Canada.  He was a postdoctoral researcher
  at ETH Zurich.  He received his Ph.D. in Electrical Engineering:
  Systems at the University of Michigan, Ann Arbor. 
  His research interests are in dynamics, control, and optimization of
  complex systems.
  He was a finalist for the 2017 Conference on Decision and Control (CDC)
  Best Student Paper Award and a recipient
  of the 2019 CDC Outstanding Student Paper Award.
\end{IEEEbiography}

\end{document}